\documentclass[a4paper,twocolumn,10.7pt,accepted=2025-12-19]{quantumarticle}
\pdfoutput=1
\usepackage[utf8]{inputenc}
\usepackage[english]{babel}
\usepackage[T1]{fontenc}
\usepackage{amssymb}
\usepackage{epsfig}
\usepackage{soul,color,epstopdf}
\usepackage{amsthm, thm-restate}
\usepackage{hyperref}
\usepackage{csquotes}
\usepackage{mathtools}
\usepackage{xcolor}
\usepackage{comment}
\usepackage{float}
\usepackage{physics}
\usepackage{dsfont}
\usepackage[caption=false]{subfig}
\usepackage{algpseudocode}
\usepackage{eufrak}
\usepackage{bbold}
\usepackage{bm}
\usepackage{accents}
\usepackage{ulem}
\usepackage{hyperref}
\usepackage{cite}

\usepackage{tikz}
\usepackage{lipsum}

\newtheorem{theorem}{Theorem}
\newtheorem*{theorem*}{Theorem}

\newtheorem{lemma}{Lemma}
\newtheorem*{lemma*}{Lemma}
\newtheorem{corollary}{Corollary}

\def\Tr{{\mathrm{Tr}}}

\def\tr{{\mathrm{tr}}}
\newcommand{\ch}{\mathcal H}
\newcommand{\cu}{\mathcal U}
\newcommand{\cg}{\mathcal G}
\newcommand{\cq}{\mathcal Q}
\newcommand{\cc}{\mathcal C}
\newcommand{\crr}{\mathcal R}
\newcommand{\cs}{\mathcal S}

\begin{document}

\title{Theory-independent randomness generation from spatial symmetries}

\author{Caroline L.\ Jones}
\email{CarolineLouise.Jones@oeaw.ac.at}
\thanks{}
\affiliation{Institute for Quantum Optics and Quantum Information,
Austrian Academy of Sciences, Boltzmanngasse 3, A-1090 Vienna, Austria}
\affiliation{Vienna Center for Quantum Science and Technology (VCQ), Faculty of Physics, University of Vienna, Vienna, Austria}
\orcid{0009-0007-2445-6001}
\author{Stefan L.\ Ludescher}
\email{Stefan.Ludescher@oeaw.ac.at}
\thanks{\newline CLJ and SLL contributed equally to this work.}
\affiliation{Institute for Quantum Optics and Quantum Information,
Austrian Academy of Sciences, Boltzmanngasse 3, A-1090 Vienna, Austria}
\affiliation{Vienna Center for Quantum Science and Technology (VCQ), Faculty of Physics, University of Vienna, Vienna, Austria}
\orcid{0000-0002-7259-1327}
\author{Albert\ Aloy}
\affiliation{Institute for Quantum Optics and Quantum Information,
Austrian Academy of Sciences, Boltzmanngasse 3, A-1090 Vienna, Austria}
\affiliation{Vienna Center for Quantum Science and Technology (VCQ), Faculty of Physics, University of Vienna, Vienna, Austria}
\orcid{0000-0002-1401-0184}
\author{Markus P.\ M\"uller}
\affiliation{Institute for Quantum Optics and Quantum Information,
Austrian Academy of Sciences, Boltzmanngasse 3, A-1090 Vienna, Austria}
\affiliation{Vienna Center for Quantum Science and Technology (VCQ), Faculty of Physics, University of Vienna, Vienna, Austria}
\affiliation{Perimeter Institute for Theoretical Physics, 31 Caroline Street North, Waterloo, Ontario N2L 2Y5, Canada}
\orcid{0000-0002-8086-5586}

\date{January 12, 2026}

\begin{abstract}
We demonstrate a fundamental relation between the structures of physical space and of quantum theory: the set of quantum correlations in a rotational prepare-and-measure scenario can be derived from covariance alone, without assuming quantum physics. To show this, we consider a semi-device-independent randomness generation scheme where one of two spatial rotations is performed on an otherwise uncharacterized preparation device, and one of two possible measurement outcomes is subsequently obtained. An upper bound on a theory-independent notion of spin is assumed for the transmitted physical system. It turns out that this determines the set of quantum correlations and the amount of certifiable randomness in this setup exactly. Interestingly, this yields the basis of a theory-independent protocol for the secure generation of random numbers. Our results support the conjecture that the symmetries of space and time determine at least part of the probabilistic structure of quantum theory.
\end{abstract}

\maketitle

\section{Introduction}
The search for physical principles, from which (elements of) quantum theory can be derived, has been a promising research direction of recent years for the foundations of physics~\cite{Hardy1,vanDam,Navascues,DakicBrukner,Chiribella,Masanes,Popescu,Mueller,Plavala}. In analogy to the derivation of the special theory of relativity from the invariance of light speed and the principle of relativity alone, the hope is that some set of simple assumptions may offer us a more profound insight into the nature of quantum theory. Here, we propose a new strategy for the reconstruction program, motivated by the interplay between quantum theory and spacetime physics.

This interplay is the subject of both historical successes and contemporary open problems. Quantum field theory was formulated by combining the principles of Minkowski spacetime with abstract quantum theory, yielding an immensely successful framework which explains almost all microphysical phenomena. Quantum gravity, on the other hand, is the contemporary problem of combining general-relativistic spacetime physics with quantum theory. These successes and open problems motivate us to ask a more general question: \textit{what is the most general consistent interplay between spacetime and probability?} In this work, we take a first, modest step in this program by determining consequences of covariance under spatial rotations around a fixed axis. As in other works within the reconstruction program, we ask the question without assuming the validity of quantum theory.

\begin{figure}[t]
\centering 
\includegraphics[trim=100 80 120 50,clip, width=1.0\columnwidth]{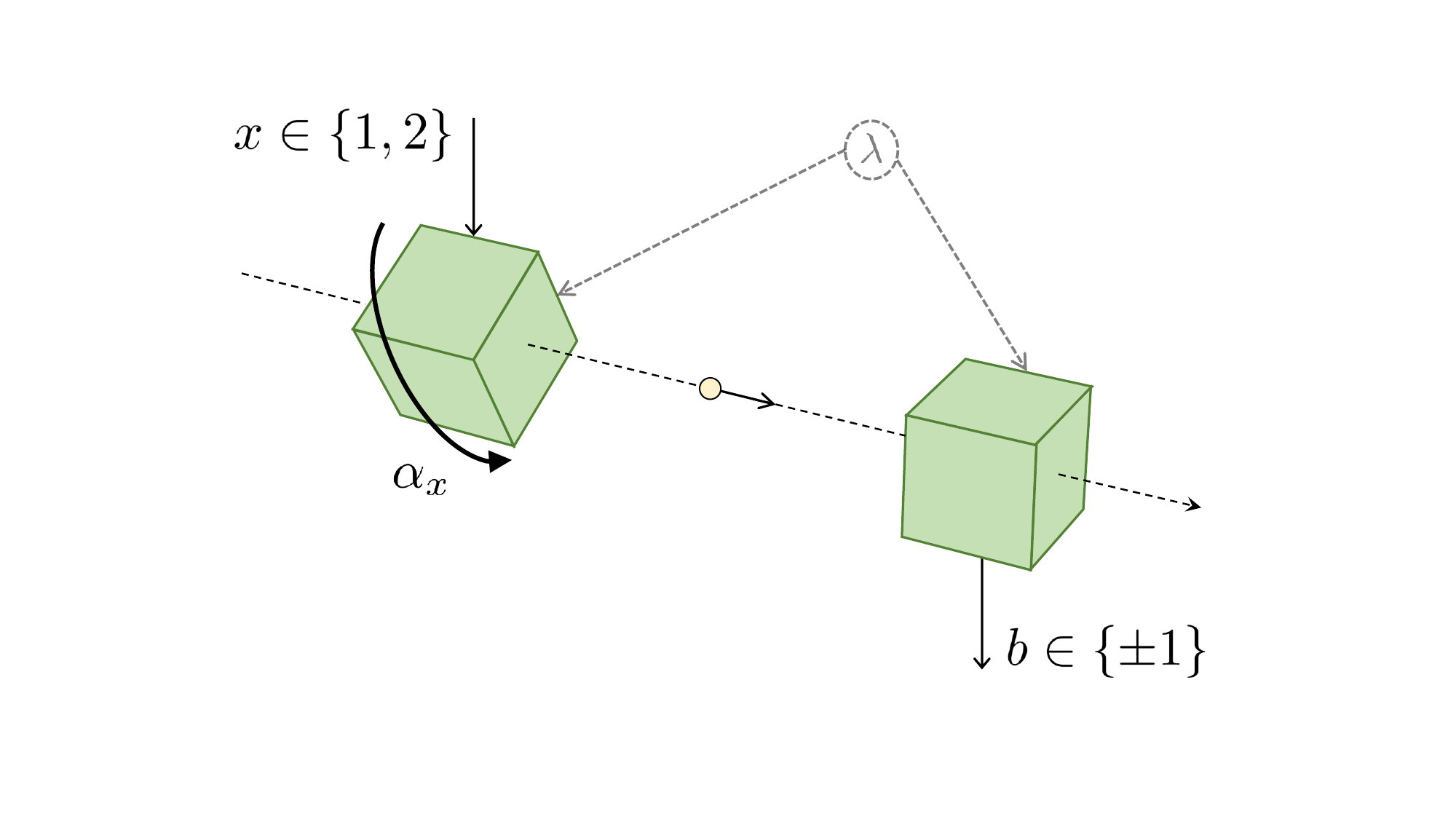}
\caption{Setup. A preparation device $P$ is rotated by an angle $\alpha_x\in\{0,\alpha\}$ relative to the measurement device $M$, and then a fixed but arbitrary state is generated. The state is then sent to $M$, where a measurement yields one of two outcomes $b\in\{\pm 1\}$. Some additional classical random variable $\lambda$, unknown to the experimenters, may have been preshared by the devices.} 
\label{figSetup}
\end{figure}

 We propose to use semi-device-independent (semi-DI) quantum information protocols as a theory-independent framework for studying this interrelation. We demonstrate that the symmetries of a scenario offer strict constraints on the probabilities admittable by nature; in particular, we show that for a simple prepare-and-measure scenario, rotational covariance alone is sufficient to recover the quantum bound on probabilities precisely if a theory-independent assumption on the transformation behavior of the transmitted systems is made.
DI and semi-DI approaches \cite{mayers1998quantum, barrett2005no, colbeck, acin2007device, gallego2010device, pawlowski2011semi, liang2011semi,branciard2012one,VanHimbeeck2017} treat devices in an experiment as ``black boxes'': no assumptions (or only very mild ones) are made about the inner workings of the devices, and the analysis relies on the observed input-output statistics alone. While Bell and other DI black-box scenarios have previously been used to study the foundations of quantum theory \cite{brunner2014bell,scarani2019bell}, here we suggest to ``put the boxes into space and time''.

Specifically, we consider the prepare-and-measure scenario sketched in Fig.~\ref{figSetup}. This setup can be used to generate random numbers~\cite{li2011semi,acin2016certified,ma2016quantum,VanHimbeeck2017,rusca2019self,tebyanian2021semi} that are secure even against eavesdroppers with additional classical information -- parametrized by a random variable $\lambda$, which may be preshared between the devices. We define a novel class of semi-DI protocols based on an assumption about how the transmitted system may respond to spatial rotations. Crucially, this semi-DI assumption does not rely on the validity of quantum theory, since it is representation-theoretic in nature and hence applies to all possible probabilistic theories. We show that the exact shape of the set of quantum correlations in this setup  can be recovered precisely from rotational covariance. Not only does our result entail the protocol's security even against post-quantum eavesdroppers, but it establishes an intimate connection between the symmetry structure of spacetime and quantum information theoretic probabilities.

\section{The setup} We consider a semi-DI random number generator similar to the one described in~\cite{VanHimbeeck2017,van2019correlations}, given by the prepare-and-measure scenario depicted in Fig.~\ref{figSetup}. The goal is to generate statistics $P(b|x)$ that certify that even external eavesdroppers with additional (classical) knowledge cannot predict $b$. As in standard DI quantum information, the security of semi-DI protocols does not require any assumptions on the inner-workings of the devices, but it requires some constraint on the physical system that is communicated between the devices \cite{gallego2010device,bowles2014certifying,VanHimbeeck2017}. This has often been implemented with a bound on the dimension of the Hilbert space of the transmitted system, restricting the communication to qubits or qutrits, as in \cite{gallego2010device,li2011semi,bowles2014certifying,brunner2008testing}, although this is arguably not very well-motivated for non-idealized physical scenarios. An alternative scheme was provided in \cite{VanHimbeeck2017,van2019correlations}, in which the mean value of some observable $H$ (such as the energy of the transmitted system) was constrained. This approach, however, requires trust in the valid characterization of the observable $H$, including, for example, the assumption of a specific gap above the ground state. In fact, the physical meaning of $H$ (say, as the generator of time translations) plays no direct role in their analysis. Whilst the formulation of~\cite{VanHimbeeck2017} is tailored towards its practical implementation, we have a more foundational focus on the relation to spatiotemporal quantities; thus, we propose modified semi-DI assumptions such that the protocol is grounded directly in properties of spacetime physics. Not only are these assumptions arguably physically well-motivated, but they can also be formulated without assuming the validity of quantum theory, as we will show below. This is in contrast to dimension bounds or assumptions on the expectation values of observables, which rely crucially on the validity of quantum theory.

\textbf{Assumptions.} Our main assumptions are based on the fact that the devices are embedded into spacetime, and spacetime admits spatial rotations. Specifically:
\begin{itemize}
\item[(i)] For any fixed value $\lambda$, the preparation device $P$ prepares some (in general unknown) state $\omega_\lambda$ of some (in general unspecified) associated physical system $S$. Moreover, for every fixed $\lambda$, it does so in a way that makes $S$ \textbf{uncorrelated (and in particular unentangled)} with other systems, including the measurement device.
\item[(ii)] The device $P$ can be \textbf{physically rotated} in space around some fixed, specified axis of rotation, without affecting other systems, and we can practically achieve this in the proposed protocol to good approximation. By the \textbf{covariance} of physical laws~\cite{Aloy}, it follows that rotations by some angle $\alpha$ must act as reversible transformations $T_\alpha$ on the unspecified state space of the associated system $S$, corresponding to a representation of ${\rm SO}(2)$.
\item[(iii)] The input $x$ is \textbf{statistically independent} of any potentially shared classical random variable $\lambda$.
\item[(iv)] Any possible adversary is only subject to \textbf{classical side information}.
\item[(v)] The representation $\alpha\mapsto T_\alpha$ of~(ii) has, when decomposed into real irreps as in~(\ref{gptrep}), highest representation label of at most $J$ (which we call the ``\textbf{generalized spin}'' of the system).
\end{itemize}
Some comments are in place. First, assumptions (i), (iii) and (iv) are standard assumptions in the semi-DI literature, even though (i) is usually not spelled out explicitly. They can be interpreted as assumptions on the causal structure of the setup. Assumption (ii) is our novel addition to the typical semi-DI scenarios and allows us to study the interrelation between observable correlations and rotational covariance. This preserves the ``black box'' paradigm of DI quantum information, whilst assuming that there are some trusted operations that we can carry out upon these boxes -- in particular, that we understand how to rotate it physically in space. In particular, this assumes that --- at least to good approximation --- we can act via rotations by some angle $\alpha$ on the preparation device, and that this action preserves the group structure: first rotating by $\alpha$ and then by $\beta$ will have the same effect on the device as rotating it by $\alpha+\beta$ (and, in particular, rotating by $\alpha$ and then by $-\alpha$ does nothing of relevance for the experiment to the device). Assumption (i) implies that we can describe $S$ as a generalized probabilistic theory system, where $T_\alpha$ of assumption (ii) must hence act as linear transformations~\cite{MuellerGarner}.

In order to have any interesting behavior at all, semi-DI frameworks must constrain the physical system $S$, which is (thought of as being) sent from the preparation to the measurement device. This role is here played by assumption (v), which we will explore and elaborate in more detail in the next section. In more standard semi-DI protocols, one often makes an assumption about the dimension of the state space (say, the Hilbert space dimension) or its information capacity. In our case, the uniquely relevant feature of $S$ is its reaction to spatial rotations, and notions of dimension or capacity do not play any direct role. This leaves essentially only one natural option: to bound the representation-theoretic properties of $T_\alpha$, as we do in (v).

Due to representation theory, finite dimension implies finite $J$. If we additionally assume the validity of quantum theory, and that the ${\rm SO}(2)$ rotations can be lifted to full ${\rm SO}(3)$ rotations in the scenario of interest, then finite Hilbert space dimension $d$ implies additionally $2J+1\leq d$, i.e.\ bounding $J$ is in some sense a weaker assumption than bounding $d$. Moreover, all known elementary particles have finite spin $J$, which gives us additional motivation to consider a bound on $J$ a natural assumption to be made. We acknowledge that the proposed bound on $J$ may not constitute any advantage for practical implementation over existing semi-DI protocols; rather, we emphasize our interest is in the foundational insights from the possibility of doing so.

\section{Quantum boxes} Let us start by describing the setup in terms of quantum theory, which we will later generalize to a theory-agnostic description. We consider two devices (Fig.~\ref{figSetup}). The first device prepares {some} quantum state $\rho_1$ and takes an input $x\in\{1,2\}$. The experimenter either does nothing to the device i.e. applies a rotation by an angle $0$, if $x=1$, or rotates it by an angle $\alpha$ {around a fixed axis} relative to the other device, if $x=2$. After the rotation, the preparation procedure is performed (say, by pressing a button on the preparation device), and the resulting physical system is sent to the second device. The second device produces an outcome $b\in\{\pm 1\}$, and is described by a POVM (positive operator-valued measure) $\{M_b\}$. Minimal assumptions are made about the devices \cite{pironio2016focus}, such that $\rho_1$ and $M_b$ are treated as unknown and may fluctuate according to some shared random variable $\lambda$.

While we allow such shared randomness (see Eq.~(\ref{eqQuantum}) below), we do not allow shared entanglement between preparation and measurement devices, which is a standard assumption in the semi-DI context~\cite{Pauwels}. Disallowing this, and demanding that the full preparation device is rotated, turns out to be necessary to obtain a resulting representation of ${\rm SO}(2)$; more on this below. For some physical intuition for what could go wrong if this was not assumed, consider a single spin-$1/2$ fermion in a superposition of two distant spatial modes, and the preparation device affecting only one of the two modes. In this case, a rotation of $2\pi$ could introduce a relative phase of $(-1)$, which could be detected by a measurement device implementing an interference experiment. This would be in direct contradiction to the results we establish below, where every $2\pi$-rotation must act trivially.

Well-known arguments (e.g.\ in~\cite[Sec.\ 13.1]{Wald}) imply that fundamental symmetries, such as the rotations by $\alpha$, must act as unitary transformations $U_\alpha$ on Hilbert space, furnishing a projective representation of the symmetry group (here ${\rm SO}(2)$). See Section~\ref{DiscussionRotBoxes} on rotation boxes (of which quantum boxes must be a subset) and~\cite[Sec.~III]{Aloy} for further detail. All finite-dimensional projective representations of ${\rm SO}(2)$ can be written in the form
\begin{equation}
{U}_{\alpha}=\bigoplus_{j=-J}^J  e^{ij\alpha}\mathbf{1}_{n_j},\label{repres}
\end{equation}
where $j$ runs over either integers or half-integers, $n_j\in \mathbb{N}_0$ denotes the multiplicity of $j$-th irrep, and $\mathbf{1}_{n_j}$ acts on the corresponding $n_j$-dimensional subspace. The assumption about the response of the system to rotations is implemented via an upper bound $J$ on the absolute value of these labels. For details see Appendix~\ref{app:projectiverepresentations}.

Fixing some $J\in\{0,\frac 1 2,1,\frac 3 2,\ldots\}$ introduces an assumption on the physical system that is sent from the preparation to the measurement device, namely, on its possible response to spatial rotations. This is what makes our scenario semi-DI, and what replaces the more common assumption on the Hilbert space dimension of the transmitted system. It is important to note that we do not fix the numbers $n_j$, thus allowing for the number of copies to vary, i.e.\ the Hilbert space dimension is not bounded by this. The number $J$ upper-bounds the spin or angular momentum quantum number associated with the physical system that is sent from the preparation to the measurement device. For example, if we have a single massive particle of spin $J$, then $U_\alpha=\exp(i \alpha Z_J)$, where $Z_J={\rm diag}(J,J-1,\ldots,-J)$ is the spin-$J$ representation of the Pauli $Z$ matrix, such that $n_j=1$ for $j=-J,-J+1,\ldots,J$. However, since the $n_j$ are arbitrary, the representation~(\ref{repres}) is allowed to be reducible, which includes the case of composite systems. For example, if the measurement probes helicity with a polarizer, then sending a single photon corresponds to a scenario with $J=1$, and $N$ photons to $J=N$~\cite{caban2003photon}. Moreover, every $J'>N$ will serve as a valid upper bound. 

As a concrete illustration, which is sketched in Fig.~\ref{figSetup2}, consider the preparation device $P$ to be constituted by a source that produces single photons on demand. Upon input $x\in\{1,2\}$, the source is either untouched (if $x=1$), or it is rotated by angle $\alpha$ around the beam axis (if $x=2$).  The measurement device $M$ is fixed in the optical axis, and could for instance be a polarizing beam splitter (PBS) with two single-photon detectors producing binary outcomes $b\in\{\pm 1\}$. Since the photon has helicity $J=1$, a linearly polarized photon will consequently have its polarization rotated by $\alpha$. On arbitrary polarization states, mechanically rotating the preparation device by $\alpha$ implements the unitary $U_\alpha={\rm diag}(e^{i\alpha},e^{-i\alpha})$ on the polarization subspace, corresponding to the spin-$1$ action assumed in Eq.~\eqref{repres} with $n_{-1}=n_1=1$ and all other $n_j=0$. Then, the detection probabilities are
\begin{equation}
    P(+1|\alpha)=\frac{1}{2}\left(1+\cos(2(\theta_0+\alpha))\right),
    \label{eqProbPhoton}
\end{equation}
where $\theta_0$ is the initial photon polarization from the source. As explained in Sec.~\ref{RBs}, this coincides with the trigonometric expressions we encounter in our rotation box analysis (c.f.~\eqref{SBs}). In a realistic scenario, one might encounter imperfections such as detector inefficiencies or multi-photon emissions which will lead to some rounds in which the $J$-bound fails. These imperfections are naturally taken into account by the relaxed set of correlations $\mathcal{Q}^\delta_{J,\alpha}$ we shall introduce later in Eq.~\eqref{relaxedquantumset}.

A mathematically equivalent probability rule will be obtained if we rotate the measurement device, rather than the preparation device, which can sometimes be more practical than rotating the single-photon source. Our group-theoretic analysis also implies that, for example, Bell experiments with pairs of photons where the settings are chosen in this way exhibit outcome probabilities that follow Malus' law, i.e.\ that are degree-two trigonometric polynomials in the local angles of the polarizers~\cite{Peres}.

\begin{figure}[t]
\centering 
\includegraphics[trim=0 50 10 0,clip, width=1.0\columnwidth]{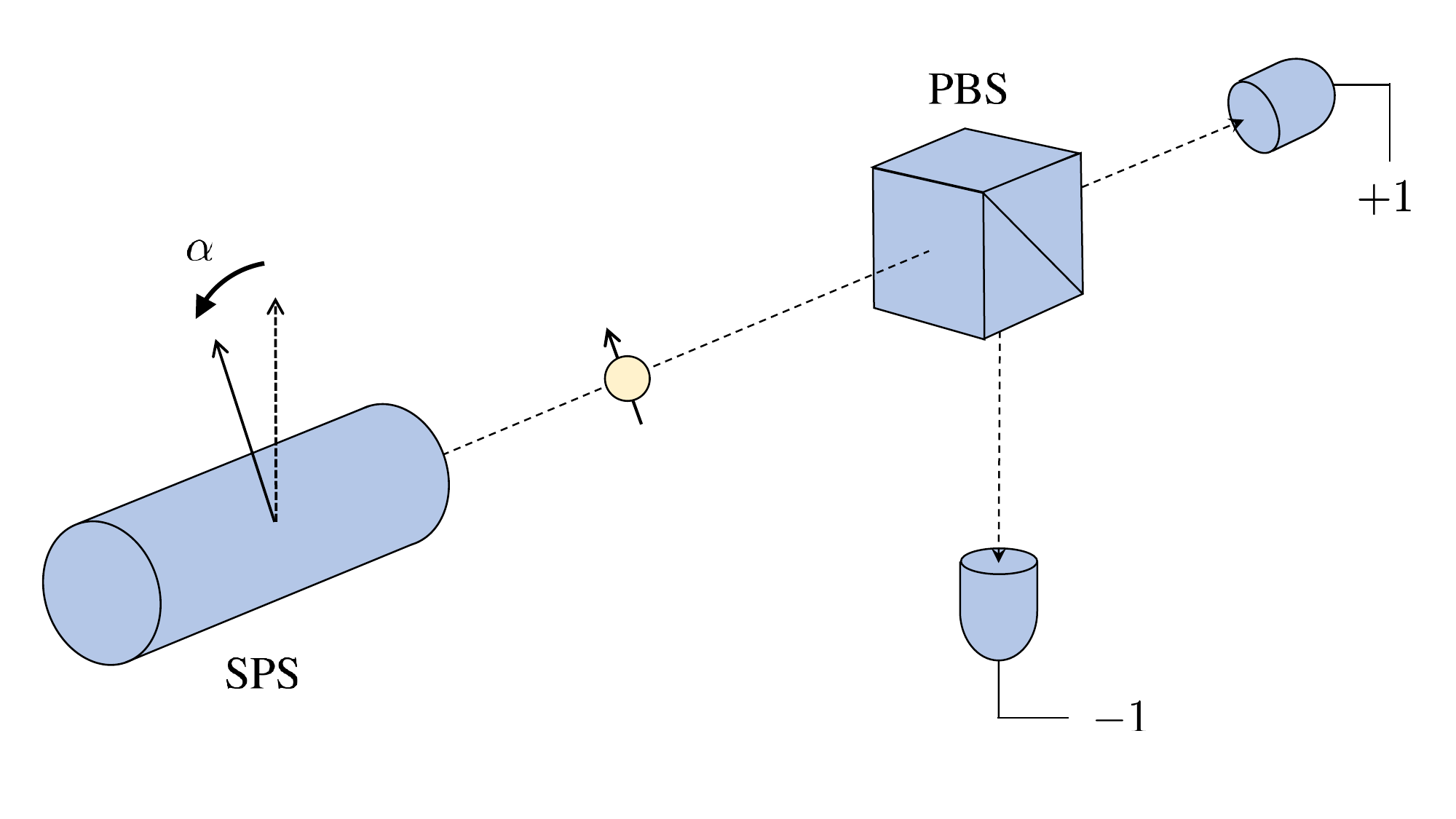}
\caption{Illustration of a possible way to implement our prepare-and-measure scenario. The preparation consists of a single-photon source (SPS) which is rotated at angles $0$ and $\alpha$ for inputs $x=1$ and $x=2$ respectively. This leaves the photon unchanged or rotates its polarization by $\alpha$. The photon is then measured by a polarizing beam splitter (PBS) with two single-photon detectors corresponding to the binary outcomes $b\in\{+1,-1\}$. This setup realizes the $J=1$ case of our framework, where the mechanical rotation implements the unitary $U_\alpha={\rm diag}(e^{i\alpha},e^{-i\alpha})$ to rotate the photon polarization.}

\label{figSetup2}
\end{figure}

Our mathematical formulation does not presuppose that the ${\rm SO}(2)$-representation must arise from spatial rotations: it could also arise for some other reason, e.g.\ from periodicity of time evolution. However, the special case that the preparation device is physically rotated in space is a paradigmatic instance in which the group symmetry is manifestly imposed from special covariance~\cite{Aloy}.

We are interested in the possible correlations between outcome $b$ and setting $x$ that can be obtained under an assumption on $J$ via Eq.~(\ref{repres}) in the quantum case. Let us for the moment assume that the initial state $\rho_1$ is a pure state $\rho_1=\ketbra{\phi_1}{\phi_1}$, then $|\phi_2\rangle=U_\alpha|\phi_1\rangle$ is prepared on input $x=2$, and the observable $M=M_1-M_{-1}$ characterizes the measurement procedure. If we consider all possible pure states $|\phi_1\rangle$ and observables $M$ arising from POVMs $\{M_b\}_{b\in\{-1,+1\}}$ in this way, then
\begin{equation}
  \cq_{J,\alpha}\coloneqq\{(E_1,E_2)|E_x=\langle \phi_x|M|\phi_x\rangle,\ket{\phi_2}=U_\alpha\ket{\phi_1}\}
\end{equation}
is the set of all possible correlations arising in our scenario, and $E_x=P(+1|x)-P(-1|x)$ is the expectation value of $M$, characterizing the bias of the outcome toward $\pm1$ for a given $x$. 
In~\cite{VanHimbeeck2017} it was shown that the set of correlations $(E_1,E_2)$ arising from arbitrary measurements on a pair of possible pure states $\phi_1,\phi_2$ is uniquely determined by their overlap $|\langle\phi_1|\phi_2\rangle|$, and the largest possible set of correlations arises when this overlap is minimal. In Appendix~\ref{boundingdelta}, we show that for our scenario the minimum overlap is
\begin{align}
\gamma=\min |\langle \phi_1|\phi_2\rangle|=\left\{\begin{matrix}
\cos(J\alpha)&\mbox{if}& |J\alpha|< \frac{\pi}{2}\\
0 &\mbox{if}&|J\alpha|\geq\frac{\pi}{2}
\end{matrix}\right..
\label{eqQuantum2}
\end{align}
Thus, according to~\cite{VanHimbeeck2017}, $\mathcal{Q}_{J,\alpha}$ is the set of correlations that satisfy the inequality
\begin{equation}
\frac{1}{2}\left(\sqrt{1+E_1}\sqrt{1+E_2}+\sqrt{1-E_1}\sqrt{1-E_2}\right) \geq \gamma.
\label{quantum_set}
\end{equation}

\begin{figure}[b]
\centering 
\includegraphics[width=.7\columnwidth]{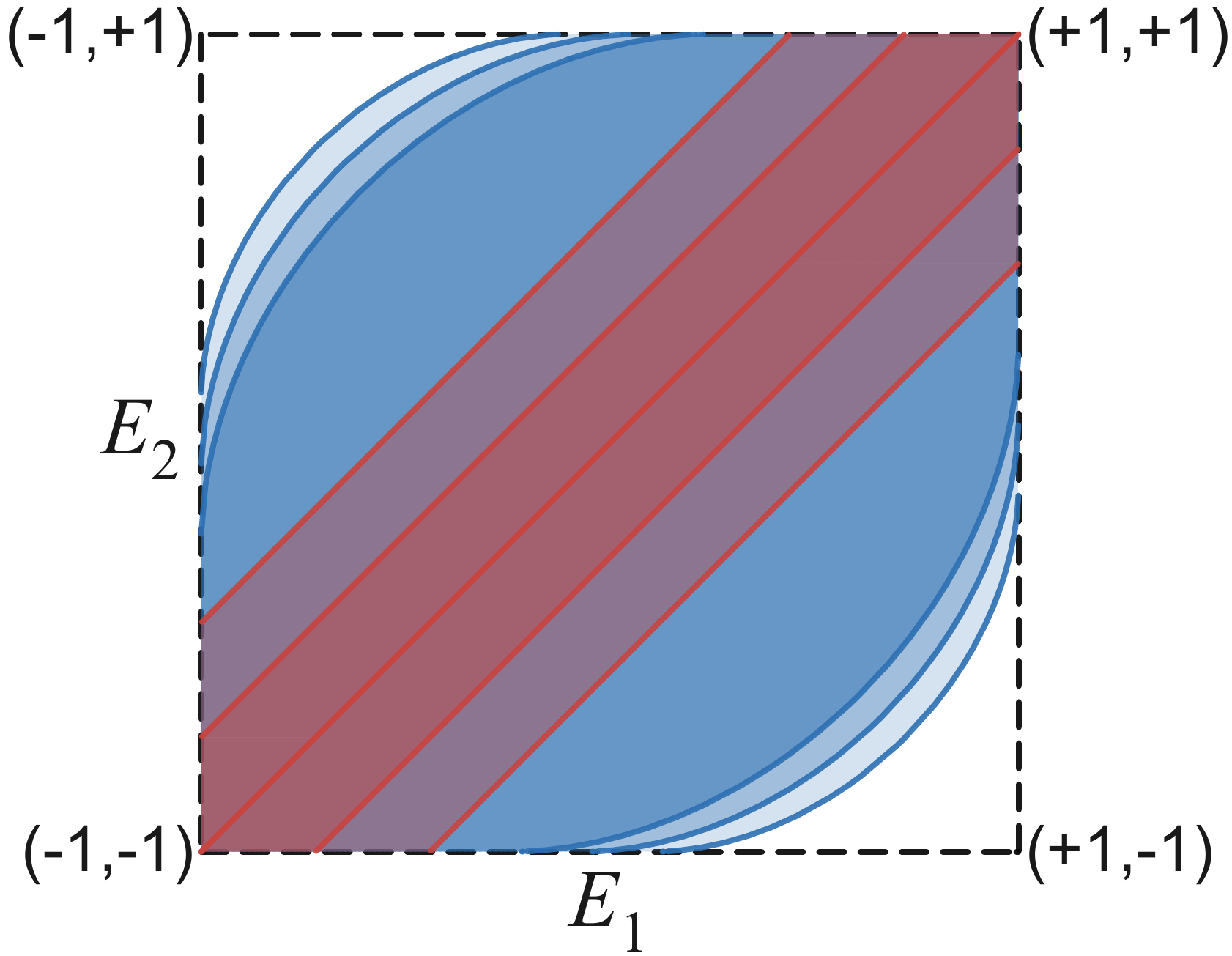}
\caption{The quantum sets $\cq_{J,\alpha}$ (dark blue) and the classical sets $\cc_{J,\alpha}$ (dark red; line $E_1=E_2$), and the quantum and classical relaxed sets $\cq^\varepsilon_{J,\alpha}$ and $\cc^\varepsilon_{J,\alpha}$ for $\varepsilon\in\{0.15,0.3\}$. We set $J=1$ and $\alpha=0.66$ in this figure.}
\label{plots}
\end{figure}

The set $\cq_{J,\alpha}$ grows with $J\alpha$ until $J\alpha=\pi/2$, at which point a $\ket{\phi_1}$ exists such that $|\phi_2\rangle=U_\alpha \ket{\phi_1}$ is orthogonal to it. If $|\phi_1\rangle$ and $|\phi_2\rangle$ are perfectly distinguishable, there exist (even deterministic) strategies to generate all conceivable correlations.

Anticipating the generation of private randomness as discussed further below, we define classical correlations as convex combinations of deterministic behaviors, i.e. $\mathbf{E}^\lambda:=(E_1,E_2)\in\{\pm 1\}\times\{\pm 1\}$,  that again satisfy the maximum spin $J$ bound: 
\begin{equation}
\cc_{J,\alpha}\!\coloneqq\!\{\mathbf{E}\!=\!\sum_\lambda p(\lambda) \mathbf{E}^\lambda\,\,|\,\,\mathbf{E}^\lambda\!\in\! \cq_{J,\alpha},\mathbf{E}^\lambda\in\{\pm 1\}\times\{\pm 1\}\},\label{classcor}
\end{equation}
where $\{p(\lambda)\}_\lambda$ is a probability distribution. If $J\alpha<\pi/2$, the states are not perfectly distinguishable, and so correlations are limited to $\mathbf{E}^\lambda=(\pm1,\pm1)$; alternatively, if $J\alpha\geq\pi/2$, the states can be perfectly distinguishable, and so $\mathbf{E}^\lambda=(\pm1,\mp1)$ are also possible correlations. Convex combinations of the former case gives the set $\cc_{J,\alpha}=\{(E_1,E_2)|-1\leq E_1=E_2 \leq 1 \}$, whilst the latter case gives all possible correlations.

So far only pure states have been considered. However, it turns out that this is sufficient, as the set of mixed state correlations, defined by
\begin{equation}
\!\! \cq'_{J,\alpha}\coloneqq\{(E_1,E_2)\,\,|\,\, E_x=\tr(M\rho_x),\rho_2=U_\alpha\rho_1 U^\dagger_\alpha\},
\end{equation} 
coincides precisely with $\cq_{J,\alpha}$. Clearly $\cq_{J,\alpha}\subseteq \cq_{J,\alpha}'$, and the converse $\cq_{J,\alpha}'\subseteq \cq_{J,\alpha}$ can be proven by purifying arbitrary states $\rho$ using an ancilla system, without adding any spin (for details, see~\ref{purestateeqmixedstateset}). Thus, the set $\cq_{J,\alpha}$ is convex, which means that it also describes scenarios where preparation $\rho_1$ and measurements $M_b$ fluctuate according to some shared random variable $\lambda$ distributed $\sim p(\lambda)$, i.e.
\begin{equation}
P(b|x)=P(b|\alpha_x)=\sum_\lambda p(\lambda){\rm tr}(M_b(\lambda)U_{\alpha_x}\rho_1(\lambda)U_{\alpha_x}^\dagger)
\label{eqQuantum}
\end{equation}
(where the input $x\in\{1,2\}$ is chosen independently from $\lambda$).
So far we have assumed that the constraint on the maximum spin $J$ holds exactly and in every run of the experiment. However, in a more realistic scenario, one may want to grant room for imperfections. This can be taken into account by trusting only that the constraint strictly holds with probability $1-\varepsilon$, with $0\leq \varepsilon< 1$, but for probability $\varepsilon$ the system might carry arbitrarily high spin. This leads to the relaxed quantum set 
\begin{equation}
\cq^\varepsilon_{J,\alpha}=(1-\varepsilon)\cq_{J,\alpha}+\varepsilon\,\left(\strut [-1,+1]\times[-1,+1]\right)
\label{relaxedquantumset}
\end{equation}  
depicted in Fig.~\ref{plots}. Similarly, if the failure of the assumption happens in such a way that it can be anticipated by a hypothetical adversary, then the set of classical strategies available to them is greater -- for which replacing $\mathcal{Q}$ by $\mathcal{C}$ in this expression defines the classical relaxed set $\mathcal{C}_{J,\alpha}^\varepsilon$. For a full characterization of the relaxed quantum and classical sets, see Appendix~\ref{relaxedsets}, where we also discuss types of experimental uncertainties for which these sets are physically relevant. For example, we show that for coherent states, where the photon number $n$ follows a Poisson distribution on Fock space, the relaxed quantum set $\cq^\varepsilon_{J,\alpha}$ with $\varepsilon=\mathcal{O}(\eta^{1/4})$ contains the relevant set of possible correlations, with $\eta:={\rm Prob}(n> N)$ giving the probability of a constraint on $J(=N)$ failing (which tends to zero exponentially in $N$).

\section{Generating private randomness} 
\label{HStar}
Adapting the results of~\cite{van2019correlations}, we can show that correlations in $\cq_{J,\alpha}$ outside of the classical set admit the generation of private randomness. Consider an eavesdropper Eve with classical (but no quantum) side information who tries to guess the value of $b$. Alice, who uses the setup of Fig.~\ref{figSetup} to generate private random outcomes $b$, will in general not have complete knowledge of all variables $\lambda\in\Lambda$ of relevance for the experiment, which is expressed in Eq.~(\ref{eqQuantum}) by $P(b|x)$ being the mixture $\sum_\lambda p(\lambda) P(b|x,\lambda)$. Eve, however, may have additional relevant information $\lambda$ (in addition to knowing the inputs $x$). It is straightforward to see that if $p(\lambda)>0$, then Eve cannot perfectly predict $b$ (i.e.\ $0<P(b|x,\lambda)<1$) if the observed correlations $\mathbf{E}$ are outside of the classical set $\mathcal{C}_{J,\alpha}$, as long as the semi-DI assumption is also satisfied for every given value of $\lambda$.

In order to generate private randomness in our scenario, Alice would like to guarantee that the conditional entropy $H(B|X,\Lambda)=-\sum_{b,x,\lambda} p(b,x,\lambda)\log_2 p(b|x,\lambda)$ is large, quantifying Eve's difficulty to predict $b$. Since $H(B|X,\Lambda)=\sum_\lambda p(\lambda) H(\mathbf{E}^\lambda)$ where $H(\mathbf{E}):=-\frac{1}{2}\sum_{b,x} \frac{1+b E_x}{2} \log \frac {1+b E_x}{2}$
, the amount of conditional entropy $H^\star$ that Alice can guarantee if she observes correlations $\mathbf{E}=(E_1,E_2)$, i.e.\ $H(B|X,\Lambda)\geq H^\star$, is determined by the optimization problem
\begin{eqnarray}
H^\star&=&\min_{\{p(\lambda),\mathbf{E}^\lambda\}} \sum_\lambda p(\lambda) H(\mathbf{E}^\lambda)\nonumber\\
&& \mbox{subject to } \sum_{\lambda:\mathbf{E}^\lambda\in\cq^{\omega}_{J,\alpha}}p(\lambda)\geq1-\varepsilon\nonumber\\
&& \mbox{and }\sum_\lambda p(\lambda)\mathbf{E}^\lambda=\mathbf{E}.
\label{optimisation}
\end{eqnarray}
That is, $H^\star$ tells us the number of certified bits of private randomness against Eve, under the assumption that the transmitted systems have spin at most $J$ --- or, rather, when this assumption holds approximately (up to some $\omega$), with high probability $(1-\varepsilon)$. This quantity is non-zero, $H^\star\equiv H^\star_{\varepsilon,\omega,\alpha}>0$, whenever the observed correlations are outside of the relaxed classical set, $\mathbf{E}\not\in\mathcal{C}^{\varepsilon}_{J,\alpha}$. For $\varepsilon=\omega=0$, this optimization problem is equivalent to the one in~\cite[Sec.\ 3.2]{van2019correlations} for the case that there is, in the terminology of that paper, no max-average assumption (see Appendix~\ref{pironiolink}). For determining the numerical value of $H^\star_{0,0,\alpha}$, we thus refer the reader to~\cite{van2019correlations}, where it is shown that $H^\star$ is the limit $\lim_{k\to\infty}H_k^\star$ of lower bounds $H^\star_k\leq H^\star$ that can be obtained via semidefinite programming. Furthermore, as we show in Appendix~\ref{boundinghstar}, we have a robustness bound for $H^\star_{\varepsilon,\omega,\alpha}$, which reads
\begin{align}
   H_{0,0,\alpha}^\star &\geq H_{\varepsilon,\omega,\alpha}^\star \nonumber\\
&\geq H^\star_{0,0,\alpha+c(\varepsilon+\omega)}+\log(1-\varepsilon)-\frac{\varepsilon\log(2/\varepsilon)}{1-\varepsilon},
\end{align}
{where $c=2\cot(J\alpha)/J$. Thus, for small $\varepsilon,\omega>0$, the number of certified random bits can still be well approximated by using the results of~\cite{van2019correlations} for $\varepsilon=\omega=0$.}

In this paper, we restrict our analysis to the characterization of the achievable correlations $\mathbf{E}$ and their possible decompositions, depending on the outcome probabilities $P(b|\alpha_x)$ for a single run of the experiment. The quantity $H^\star$ above quantifies the unpredictability of the outcome of a single round under our theory-independent assumption of a spin-$J$ bound. Implementing a concrete randomness generation protocol that aims at producing a large number of random bits with some number $n$ of rounds will in general need a more elaborate analysis, which is beyond the scope of this paper. We direct the interested reader towards Ref.~\cite{van2019correlations} for more precise details on an implementation of the analogous protocol~\cite{VanHimbeeck2017}. 

In particular, we expect that one will not need to assume that all runs of the experiment lead to independently and identically distributed settings and outcomes (i.i.d): it is now well-known that DI and semi-DI randomness generation and expansion protocols can avoid the i.i.d.\ assumption by using the martingale framework and the Azuma–Hoeffding inequality (see e.g.,~\cite{grimmett2020probability,pironio2013security,zhang2018certifying, knill2020generation, van2019correlations}). Here we follow the prepare-and-measure framework of~\cite{van2019correlations}, which builds on this martingale-based approach and is thus secure without requiring i.i.d.\ assumptions, allowing the devices to have memory or to vary their behaviour between rounds, while still requiring that the inputs $x$ are chosen independently from the past and in an identically distributed way. In particular, the underlying assumptions for an actual implementation of the protocol will be those of~\cite[Subsection 4.1]{van2019correlations}, where the ``max-average assumption'' is dropped, and the ``max-peak assumption'' is replaced by the assumption that $\mathbf{E}_i(W_i,\Lambda)\in\mathcal{Q}_{J,\alpha}$ for all $1\leq i\leq n$, i.e.\ that the correlation $\mathbf{E}_i$ of every $i$-th run, which may depend on the variables $W_i$ encoding the inputs and outputs of the previous runs and the shared randomness $\Lambda$, is a valid spin-$J$ correlation.

\section{Rotation boxes} \label{RBs}
We now drop the assumption that quantum theory holds, and consider the most general form {the probabilities $P(b|\alpha)$ may take that is} consistent with the rotational symmetry of the setup while implementing our spin bound. As discussed later in more detail in Section~\ref{DiscussionRotBoxes}, to every prepare-and-measure scenario, we can associate a convex set of states in a real vector space (in quantum theory, these are the density operators in the space of Hermitian matrices). Covariance implies~\cite{Wald} that spacetime symmetries (and hence their subgroup ${\rm SO}(2)$) have linear representations on this space, necessarily characterized by a maximum charge (``spin'') $J$. As shown later in Section~\ref{DiscussionRotBoxes}, it follows that
\begin{equation}
P(b|\alpha)=\sum_{k=0}^{2J} \left(c_k^{(b)} \cos(k\alpha)+s_k^{(b)}\sin(k\alpha)\right),
\label{SBs}
\end{equation}
with suitable coefficients $c_k^{(b)}$, $s_k^{(b)}$. In quantum theory in particular, if $P(b|\alpha)$ satisfies Eq.~(\ref{eqQuantum}) and $U_\alpha$ satisfies the semi-DI assumption Eq.~(\ref{repres}), then it is of the form~(\ref{SBs}). Conversely, we show~\cite{Aloy} that every ``rotation box'' $P(b|\alpha)$ of the form~(\ref{SBs}), yielding valid outcome probabilities between $0$ and $1$, comes from a representation of ${\rm SO}(2)$ on some (in general non-quantum) probabilistic theory with  maximal spin $J$. Since $P(+1|\alpha)+P(-1|\alpha)=1$, the set of possible spin-$J$ quantum and rotation boxes respectively can be denoted by
\begin{eqnarray}
\cq_J&:=&\{\alpha\mapsto P(+1|\alpha)\,\,|\,\, P(b|\alpha)\mbox{ of form }(\ref{eqQuantum})\},\\
\mathcal{R}_J&:=&\{\alpha\mapsto P(+1|\alpha)\,\,|\,\, P(b|\alpha)\mbox{ of form }(\ref{SBs})\},
\end{eqnarray}
where, for $\cq_J$, we assume that $U_\alpha$ is of the form~(\ref{repres}). We have just seen that $\cq_J\subseteq \mathcal{R}_J$. Trivially, $\cq_0=\mathcal{R}_0$ is the set of constant probability functions, and it can be shown that $\cq_{1/2}=\mathcal{R}_{1/2}$ (see Appendix~\ref{proofq1/2=r1/2}). However, in~\cite{Aloy}, we show that $\cq_J\subsetneq \mathcal{R}_J$, i.e.\ that there are more general ways to respond to spatial rotations than allowed by quantum theory, if $J\geq 3/2$.

\subsection{The physical meaning of rotation boxes}
\label{DiscussionRotBoxes}

Natural extensions of quantum theory are often phrased within the language of generalized probabilistic theories (GPTs)~~\cite{Hardy1,Barrett,Mueller,Plavala}. In particular, all possible consistent statistical descriptions of prepare-and-measure scenarios can be described by a GPT system~\cite{MuellerGarner}. A GPT system $A$ consists of vector space $V_A$ (here taken to be finite-dimensional), a state space $\Omega_A\subset V_A$, an effect space $\mathcal{E}_A\subset V_A^*$ and a set of transformations $\mathcal{T}_A\subset\mathcal{L}(V_A)$. In summary, preparation procedures are described by states $\omega\in\Omega_A$, outcomes of measurements by effects $e\in\mathcal{E}_A$, and transformations $T$ by linear maps on $T\in\mathcal{T}_A$ such that $(e,T\omega)$ is the probability to obtain the corresponding outcome, following the preparation and transformation procedures. Quantum systems over $\mathbb{C}^n$ are special cases of GPT systems, with $V_A$ the space of Hermitian complex $n\times n$ matrices, $\Omega_A$ the set of density matrices, $\mathcal{E}_A$ the set of POVM elements, and $\mathcal{T}_A$ the completely positive trace-preserving maps.

Now, assuming the rotational covariance of physics, similar arguments as in the quantum case~\cite{Wald} imply that there must be a representation of ${\rm SO}(2)$ on the state space. First considering QT, the most general representation acting on a Hilbert space is given by Eq.~(\ref{repres}), which induces a representation on the real vector space of Hermitian matrices (and therefore also on the state space of density matrices) $\mathcal{U}_\alpha(\rho):=U_\alpha \rho U_\alpha^\dagger$. In a suitable basis, the superoperator $\mathcal{U}_\alpha$ has the block matrix form~\cite{Aloy}
\begin{equation}
    \mathcal{U}_\alpha=\mathbf{1}\oplus\bigoplus^{2J}_{k=1}\mathbf{1}_{m_k}\otimes\begin{pmatrix}
        \cos{(k\alpha)} & -\sin{(k\alpha)} \\
        \sin{(k\alpha)} & \cos{(k\alpha)}
    \end{pmatrix},\label{gptrep}
\end{equation}
where $m_k\in \mathbb{N}_0$ denotes the multiplicity of the $k$th real irrep.
This must be true because every representation of ${\rm SO}(2)$ on a finite-dimensional real vector space is of this form. To say that a quantum system carries a representation of ${\rm SO}(2)$ of this form, with $m_{2J}\neq 0$, is equivalent to saying that the all outcome probabilities ${\rm tr}(M U_\alpha \rho U_\alpha^\dagger)$ are trigonometric polynomials in $\alpha$ (i.e.\ of the form~(\ref{SBs})), and the maximal degree over all states $\rho$ and POVM elements $M$ is equals $2J$. These are two equivalent ways of saying that we have a quantum spin-$J$ system.

We can now drop the assumption that quantum theory holds, and say that a GPT system is a spin-$J$ system if it carries a representation of ${\rm SO}(2)$ as transformations $T_\alpha$ that can be decomposed as in~(\ref{gptrep}). Similarly as in quantum theory, this is equivalent to saying that all outcome probabilities $(e,T_\alpha\omega)$ are trigonometric polynomials in $\alpha$, of maximal degree $2J$. Moreover, all  rotation box probabilities $P(b|\alpha)$ of Eq.~(\ref{SBs}) can be seen as arising from some spin-$J$ GPT system~\cite{Aloy}.

The ``post-quantum number'' $J$ behaves in similar ways as its quantum counterpart. For example, placing two independent rotation boxes $P_1$ and $P_2$ side by side gives a resulting box $P(b_1,b_2|\alpha):=P_1(b_1|\alpha)P_2(b_2|\alpha)$ with $J=J_1+J_2$. This is in line with particle physics intuition by hinting at $J$ being related to the number of constituents or ``size'' of the physical system.

\subsection{Agreement of correlation sets} 
If we consider only two possible inputs, $x\in\{1,2\}$, with corresponding rotations by $0$ and $\alpha$ (which is a fixed angle), the resulting set of rotation box correlations is
\begin{eqnarray}
\!\!\crr_{J,\alpha} \!\! &\coloneqq& \!\! \{(E_1,E_2)|E_1\!=\!P(+1|0)\!-\!P(-1|0),\nonumber\\
&& E_2\!=\!P(+1|\alpha)\!-\!P(-1|\alpha),P\,\mbox{as in}\,(\ref{SBs})\}.
\end{eqnarray}
Obviously $\mathcal{Q}_{J,\alpha}\subseteq\mathcal{R}_{J,\alpha}$, but we can say more:
\begin{theorem}
For every fixed angle $\alpha$, the quantum set coincides with the rotation box set, i.e. $\cq_{J,\alpha}=\mathcal{R}_{J,\alpha}$.
\label{weakconj}
\end{theorem}
\begin{proof}
Clearly $\mathcal{Q}_{J,\alpha}\subseteq \mathcal{R}_{J,\alpha}$. We use~\cite[Chapter 4, Thm.\ 1.1]{DeVore}: If $T$ is a trigonometric polynomial of degree $n$ with $-1\leq T(x)\leq 1\; \forall x$, then
\begin{equation}
T'(x)^2 + n^2T(x)^2\leq n^2 .\label{trigineq}
\end{equation}
Suppose that $P$ defines some spin-$J$ rotation box correlation, i.e.\ $P(+|\alpha)$ is a trigonometric polynomial of degree at most $2J$, taking values in the interval $[0,1]$ for all $\alpha$. Define $T(\alpha)\coloneqq P(+|\alpha)-P(-|\alpha)=2P(+|\alpha)-1$, which is a trigonometric polynomial of degree $n=2J$ with $-1\leq T(\alpha) \leq 1$. Rewrite~(\ref{trigineq}) as $T'(x)\leq 2J\sqrt{1-T(x)^2}$ and set $E_x:=T(0)$ and $E_{x'}:=T(\alpha_\varepsilon)$, then
\begin{align}
    \alpha_\varepsilon = \int_0^{\alpha_\varepsilon}d\alpha &\geq \int_0 ^{\alpha_\varepsilon} \frac{T'(\alpha)d\alpha}{2J\sqrt{1-T(\alpha)^2}}\nonumber\\    
&=\frac{1}{2J}\int_{E_{x}}^{E_{x'}}\frac{dy}{\sqrt{1-y^2}}\nonumber\\
&= \frac{1}{2J}\left(\arcsin E_{x'}-\arcsin E_x\right),
\end{align}
where we have substituted $y=T(\alpha)$. It follows that
\begin{equation}
    \frac 1 2 |\arcsin E_2-\arcsin E_1|\leq J\alpha.
    \label{rotation_box_set}
\end{equation}
For $J\alpha\geq\pi/2$, the set $\mathcal{R}_{J,\alpha}$ contains all possible correlations, as in the quantum case. For $J\alpha<\pi/2$, taking the cosine of both sides of~(\ref{rotation_box_set}) reproduces, after some elementary manipulations (see Appendix~\ref{proofquantumeqrotationset}), precisely the conditions of the quantum set, as in~(\ref{quantum_set}) and~(\ref{eqQuantum2}), hence $(E_1,E_2)\in\mathcal{Q}_{J,\alpha}$, and so $\mathcal{R}_{J,\alpha}\subseteq \mathcal{Q}_{J,\alpha}$.
\end{proof}

This shows that the set of quantum correlations in our setup can be understood as a consequence of the interplay of probabilities and spatial symmetries, without assuming the validity of quantum theory. Notably, \cite{tavakoli2022informationally} also identify a general polytope $\mathcal{G}$ that characterizes the set of correlations under an abstract informational restriction~\cite{tavakoli2020informationally} when no assumption is made on the underlying physical theory, and in the simplest case of two inputs, this polytope agrees with the set of achievable quantum correlations. Here, however, we show that a physically well-motivated assumption reproduces the curved boundary of the set of quantum correlations exactly, for all $J$. Moreover, Theorem 1 implies that the amount of certifiable randomness $H^\star$ remains correct even if the validity of quantum theory is not assumed. To the best of our knowledge, there has not been a description of a semi-DI prepare-and-measure scenario with this property in earlier work.

\section{Post-quantum security} 
The equality $\mathcal{R}_{J,\alpha}=\cq_{J,\alpha}$ implies that the semi-DI protocol above is secure against post-quantum eavesdroppers. While Alice observes quantum correlations $\bm{E}\in\mathcal{Q}_{J,\alpha}$, i.e.\ of the form~(\ref{eqQuantum}), it is conceivable that these are actually mixtures of beyond-quantum rotation boxes $\bm{E}^\lambda\in \mathcal{R}_{J,\alpha}^\omega$ such that $\bm{E}=\sum_\lambda p(\lambda)\bm{E}^\lambda$, where
Eve may have access to beyond-quantum physics and know the value of $\lambda$. To see how many bits of private randomness $H^\star$ Alice can guarantee against Eve in this case, the optimization problem~(\ref{optimisation}) has to be altered by relaxing the condition on $\mathbf{E}^\lambda\in\mathcal{Q}_{J,\alpha}^\omega$ to $\mathbf{E}^\lambda\in\mathcal{R}_{J,\alpha}^\omega$, i.e.\ by only demanding that every transmitted system is, up to probability $\varepsilon$, approximately a (not necessarily quantum) rotation box of maximal spin $J$. However, since $\mathcal{R}_{J,\alpha}^\omega=\mathcal{Q}_{J,\alpha}^\omega$, the optimization problem and hence $H^\star$ are unaffected by this.

We emphasize that our aim here is not to propose a practical implementation, but to show that the essential features of the randomness protocol framework of~\cite{van2019correlations} can be derived directly from covariance constraints via the spin-$J$ bound.

\section{Conclusions} 
We have introduced a theory-independent and semi-device-independent scenario for generating random numbers based on the response of physical systems to spatial rotations. This allowed us to recover the exact set of quantum correlations of the setup without assuming quantum mechanics, merely from a semi-DI assumption on a generalized notion of spin of the transmitted system. From a fundamental point of view, our results demonstrate that the symmetries of space and time enforce important features of quantum theory in some scenarios. From a more pragmatic point of view, they allow us to certify random numbers from physically better motivated assumptions than the usual dimension bounds, and the amount of secure random bits is independent of the validity of quantum physics.

Our results demonstrate that semi-DI scenarios can shed light on the foundations of physics. In particular, they allow us to study the question of how the structure of space and time constrains the probabilities and correlations of preparation procedures and measurement outcomes. Clearly, our work has only addressed the simplest case of this problem. What else can we learn by studying more general scenarios beyond prepare-and-measure, and more general symmetry groups such as the full rotation group ${\rm SO}(3)$, time translations, or the Lorentz group? In addition, our work suggests an interesting foundational question: is quantum theory perhaps the only probabilistic theory that ``fits into space and time'' for all scenarios? A positive answer to this question would significantly improve our understanding of the logical architecture of our world. On the other hand, a negative answer could inform experimental tests of quantum theory, by telling us where there might be elbow room for beyond-quantum physics consistent with spacetime as we know it.\\

\section*{Acknowledgments} 
We are grateful to Valerio Scarani and Armin Tavakoli for helpful discussions. We acknowledge support from the Austrian Science Fund (FWF) via project P 33730-N. This research was supported in part by Perimeter Institute for Theoretical Physics. Research at Perimeter Institute is supported by the Government of Canada through the Department of Innovation, Science, and Economic Development, and by the Province of Ontario through the Ministry of Colleges and Universities.

\onecolumn 
\bigskip
\appendix

\section*{Appendix}

\section{Bounding the overlap $\gamma$}
\label{boundingdelta}
We have to determine $\gamma:=\min|\langle\phi|U_\alpha|\phi\rangle|$, where the minimization is over all representations of the form~(\ref{repres}) and over all pure states $|\phi\rangle$ in all finite-dimensional Hilbert spaces supporting such representations. Here we will show that
\begin{align}
\gamma=\left\{
   \begin{array}{cl}
      \cos(J\alpha) & \mbox{if }|J\alpha|<\frac\pi 2 \\
      0 & \mbox{if }|J\alpha|\geq\frac \pi 2
   \end{array}
\right\}.
\end{align}
That this bound is attained, i.e.\ that the right-hand side upper-bounds $\gamma$, can be seen by considering the states $(|j\rangle+e^{i\theta}|-j\rangle)/\sqrt{2}$ on $\mathbb{C}^{2J+1}$ carrying the spin-$J$ irrep of ${\rm SO}(2)$. Every normalized state $|\phi\rangle$ can be written in the form $|\phi\rangle=\sum_{j=-J}^J \phi_j|\psi_j\rangle$, where $\phi_j|\psi_j\rangle:=\mathbf{1}_{n_j} |\phi\rangle$, $\sum_j |\phi_j|^2=1$, and the $|\psi_j\rangle$ are normalized and satisfy $U_\alpha |\psi_j\rangle=e^{ij\alpha}|\psi_j\rangle$.

To prove the bound above, we evaluate the inner product with the rotated state $U_\alpha\ket{\phi}=\sum_k e^{ij\alpha}\phi_j\ket{\psi_j}$. Using that $|z|\geq{\rm Re}(z)$ for $z\in\mathbb{C}$, we obtain
\begin{align}
    \abs{\bra{\phi}U_\alpha\ket{\phi}}&=\abs{\sum_j|\phi_j|^2e^{ij\alpha}}\nonumber\\
    &=\abs{ |\phi_0|^2+\sum_{j>0}\big(|\phi_{-j}|^2e^{-ij\alpha}+|\phi_{j}|^2e^{ij\alpha}\big)}\nonumber\\
    &\geq|\phi_0|^2+\sum_{j>0}\cos(j\alpha)\big(|\phi_{-j}|^2+|\phi_{j}|^2\big).
\end{align}
There is a factor of $\cos(j\alpha)$ in front of all coefficients, which is smaller for larger $j$, i.e.\ $\cos(j_1 \alpha)<\cos(j_2 \alpha)$ for $j_1>j_2$ since $|\alpha|\leq\pi/(2 J)$. Therefore the final expression is minimized when the coefficients are weighted entirely by the maximum $j$ terms, i.e.\ for $|\phi_{-J}|^2+|\phi_J|^2=1$, and so $\gamma\geq\cos(J\alpha)$.

\section{Projective Representations of ${\rm SO}(2)$}
\label{app:projectiverepresentations}

In this section, we analyze the representations of ${\rm SO}(2)$. The authors do not claim originality in finding all projective representations of SO$(2)$, but they were unable to find a source presenting them in the desired form. The first part of this appendix is particularly technical and intended mainly for interested readers. Those who prefer may skip directly to (\ref{eq:projectiverepresentations}), which provides the most general form of projective representation, or to (\ref{finalformrepresent}), which gives the most general representation required in the main text. 

We will restrict our discussion to representations on finite-dimensional Hilbert spaces. 
In general, we consider projective unitary representation of ${\rm SO}(2)$. In finite dimension it is always possible to deprojectivize unitary representations of the connected Lie group $\cg$ by passing to unitary representations of its universal cover $\hat{\cg}$, and moreover, every projective unitary representation of $\mathcal{G}$ stems from a ordinary unitary representation of $\hat{\mathcal{G}}$ \cite{hall2013}. It can be checked that the universal cover of ${\rm SO}(2)$ is the translation group $(\mathbb{R},+)$ \cite{hatcher2002, hall2015}. Using Stone's theorem \cite{stone1932}, we find that the projective representations of ${\rm SO}(2)$ must be (up to a global phase) of the form
\begin{equation}
    U_\alpha=\exp(i\,\mbox{diag}(\begin{matrix}
        j_1 ,& \ldots ,&j_N
    \end{matrix})\alpha),
\end{equation}
where all $j_i\in\mathbb{R}$, $\alpha\in\mathbb{R}$, and $U_{2\pi n}=e^{i\phi_n}\mathbf{1}$ with $\phi_n\in \mathbb{R}$.

From $U_{2\pi}=\exp(i\phi_1 2\pi)\mathbf{1}$, it follows that\begin{equation}
    j_i=\phi_1+k  \label{spinsinrep}
\end{equation}  
where $k\in \mathbb{Z}$.
It is easy to check that every $U_\alpha$ of the form
\begin{equation}
U_\alpha=e^{i\phi_1\alpha}\bigoplus_{k\in\mathbb{Z}} e^{ik \alpha} \mathbf{1}_{n_k},\label{eq:projectiverepresentations}
\end{equation}
where $\phi_1\in\mathbb{R}$, $n_k\in \mathbb{N}$, $\mathbf{1}_{n_k}$ acts on the corresponding $n_k$-dimensional subspace $\ch_k$, and only a finite number of the $n_k\neq 0$, gives a valid projective representation of ${\rm SO}(2)$. 

At first glance, it seems like that $\phi_1$ introduces only an (on $\alpha$ depending) global phase, which cannot be observed and thus one might conclude that w.l.o.g  we can set $\phi_1=0$. However, it turns out that $\phi_1$ plays an important role as soon as we implement our constraint $|j_i|\leq J$. An immediate consequence of (\ref{spinsinrep}) is that $j_i-j_k\in\mathbb{Z}$ for all $i,k$. Hence, $|j_i|\leq J$ will constrain the number of allowed $j_k$ (up to multiplicity $n_k$). Namely, the maximal number of different $j_k$'s is given by $2\tilde{J}+1$, where $\tilde{J}\in \mathbb{N}_0\cup \mathbb{N}/2$ and $\tilde{J}\leq J< \tilde{J}+1$. This implies that we can restrict to projective representations of the form \begin{equation}
    U_\alpha=\bigoplus_{j=-J}^J e^{ij\alpha} \mathbf{1}_{n_j},\label{finalformrepresent}
\end{equation}   
where $J \in \mathbb{N}_0\cup \mathbb{N}/2$ and all $j\in \mathbb{N}_0$ (i.e. $\phi_1=0$) or all $j\in \mathbb{N}/2$ (i.e.\ $\phi_1=-1/2$) depending on if $J$ is in $\mathbb{N}_0$ or $\mathbb{N}/2$.

For a better understanding why we only need to consider representations of the form (\ref{finalformrepresent}), we look at an example where  $J=3/4$ and $j_1=-1/4$ and $j_2=3/4$. Thus we have a representation \begin{equation}
    U_\alpha=e^{-i \frac{\alpha}{4}}\mathbf{1}_{n_1}\oplus e^{i\frac{3\alpha}{4}}\mathbf{1}_{n_2}=e^{i\frac{\alpha}{4}}\left(e^{-i \frac{\alpha}{2}}\mathbf{1}_{n_1}\oplus e^{i\frac{\alpha}{2}}\mathbf{1}_{n_2}\right).
\end{equation}
However, we can define a new projective representation by $\tilde{U}_\alpha=e^{-i\frac{\alpha}{4}}U_\alpha$, which is now of the form (\ref{finalformrepresent}) with maximum $\tilde{J}=1/2$ and both representations will lead to the same observable physics.

So far we only considered representations on Hilbert spaces and thus on pure states. Now, we want to argue that this is actually sufficient. In principle, one could think of situations where in every run the measurement device 
picks a system associated with a different Hilbert space $\ch^i$ with a different maximum $J^i$ and a different phase $\phi^i$ appearing in $U^i_\alpha$. More generally, we can have physical systems that feature incoherent mixtures of bosonic and fermionic degrees of freedom, in accordance with an univalence superselection rule~\cite{Wightman}. Hence, the most general state space to consider is $\cs\left(\bigoplus_i\ch^i\right)$. Then, we have a representation $\cu_\alpha$ acting on $\mathcal{B}\left(\bigoplus_i\ch^i\right)$, which can be written as 
\begin{equation}
    \cu_\alpha(\rho)=\bigoplus_i (U^i_\alpha)^\dagger \bigoplus_k\rho^k \bigoplus_lU^l_\alpha=\bigoplus_i (U^i_\alpha)^\dagger \rho^i U^i_\alpha ,
\end{equation}
where $U^i_\alpha$ is a representation of ${\rm SO}(2)$ on $\ch^i$ and $\rho^i$ is a subnormalized state acting on $\ch^i$ and we can interpret $p_i=\Tr(\rho^i)$ 
 as the probability that the box prepares a state $\widetilde{\rho}^i =\rho^i/\Tr(\rho^i)\in\cs(\ch^i)$. For a POVM element $M_+$ the probabilities are given by \begin{eqnarray}
    P(+|\alpha)&=&\Tr(\cu_\alpha(\rho)M_+)=\sum_i \Tr((U^i_\alpha)^\dagger\rho^iU^i_\alpha M^i_+ )=\sum_ip_i P_i(+|\alpha),\label{probabiliteisfordifferentJm}
\end{eqnarray}
where $M^i_+=\Pi^iM_+\Pi^i$, $\Pi^i$ is the projection on $\mathcal{H}^i$ and $P_i(+|\alpha)=\Tr((U^i_\alpha)^\dagger\widetilde{\rho}^iU^i_\alpha M^i_+)\in \cq_{J^i}$. Now let $\ch^m$ be the subspace with the highest maximal spin i.e. $J^i\leq J^m$, for all $i$. This implies $\cq_{J^i}\subseteq \cq_{J^m}$ for all $i$, and in particular, we have $P_i(+|\alpha)\in\cq_{J^m}$. Thus, probabilities $P(+|\alpha)$ of the form~(\ref{probabiliteisfordifferentJm}) are convex combinations of elements of $\cq_{J^m}$ and hence $P(+|\alpha)\in\cq_{J^m}$. It follows that we could have found the same correlations only by considering the representation $U^m_\alpha$. This implies that we can restrict to representations of the form (\ref{finalformrepresent}) without loss of generality.

\section{Proof that ${\mathcal{Q}_{J,\alpha}=\mathcal{Q}'_{J,\alpha}}$}
\label{purestateeqmixedstateset}

We can show $\mathcal{Q}_{J,\alpha}\supseteq \mathcal{Q}'_{J,\alpha}$ by considering the purification of mixed states, and ensuring that the purification is carried out in such a way that it does not add any extra spin.

In purifying the state, we are embedding our Hilbert space $\mathcal{H}_{A}$ into a larger Hilbert space $\mathcal{H}_{AB}=\mathcal{H}_{A}\otimes\mathcal{H}_{B}$, where $\text{dim}(\mathcal{H}_A)=\text{dim}(\mathcal{H}_B)$. We require that $\mathcal{H}_{AB}$ has the same constraint on $J$, in order that it doesn't give any new correlations. Since the total spin of the composite system is the sum of the spins of the two individual systems, the Hilbert space of the ancilla system $\mathcal{H}_{B}$ must carry only the trivial representation. 

Let $\ket{\psi_1}$ be a purification of $\rho_1$, and let $\ket{\psi_2}:=U_A(\alpha)\otimes\mathbf{1}_B\ket{\psi_1}$, then $\ket{\psi_2}$ is also a purification of $\rho_2$. Since $\tr(\rho_x M)=\tr(\ketbra{\psi_x}{\psi_x} M_A\otimes\mathbf{1}_B)$, we have realized the mixed-state correlation via pure states. Thus, for any given correlation realized by mixed states $(E_1,E_2)\in\mathcal{Q}'_{J,\alpha}$, this correlation can also be realized via a pure state with the same bound $J$, i.e.\ $(E_1,E_2)\in\mathcal{Q}_{J,\alpha}$.

\section{Relaxed quantum \& classical sets}
\label{relaxedsets}

\textit{Relaxed quantum set.} We now characterize a more general quantum set, for which we only probabilistically know $J$. We consider the set of correlations such that the probability that $J_{\lambda}\leq J$ is at least $1-\omega$, where $0\leq \omega<1$:
\begin{equation} \label{probabilistic J defn}
   \mathcal{Q}^\omega_{J,\alpha}\coloneqq\left\{\left.\mathbf{E}\!=\!\sum_{\lambda} p(\lambda) \mathbf{E}_{\lambda}\,\,\right|\,\, \mathbf{E}_{\lambda}\!\in\!\mathcal{Q}_{J_{\lambda},\alpha},\sum_{\lambda:J_\lambda\leq J} p(\lambda)\geq 1\!-\!\omega\right\},
\end{equation}
where the second sum ranges over the values of $\lambda$ for which $J_\lambda\leq J$. We have labeled this first type of error $\omega$, in order to signify that it refers to a genuine, \textit{ontic} randomness, such as that associated with e.g.\ bounding the photon number when measuring single-mode coherent states. Note that the {correlations} $\mathbf{E}$ are convex mixtures (weighted by the probabilities) of $\mathbf{E}_\lambda$. Therefore we claim that the probabilistic correlations are of the following form:
\begin{equation} \label{probabilistic J form}
   \mathcal{Q}_{J,\alpha}^\omega=\{(1-\omega)\mathbf{E}+\omega\mathbf{E}'\,\,|\,\, \mathbf{E}\in\mathcal{Q}_{J,\alpha},\enspace \mathbf{E}'\mbox{ arbitrary}\}.
\end{equation}

We would like to show that the correlations $\bm{E}^\omega$ in $\mathcal{Q}_{J,\alpha}^\omega$, as defined in Eq.~(\ref{probabilistic J defn}), are of the form of Eq.~(\ref{probabilistic J form}), i.e. they are of the form:
\begin{align}\label{probabilistic J eqn}
\bm{E}^\omega\coloneqq&\sum_\lambda p(\lambda)\bm{E}_\lambda=(1-\omega)\bm{E}^J+\omega\bm{E}',
\end{align}
where $\mathbf{E}^J\in\mathcal{Q}_{J,\alpha}$, and $\mathbf{E}'$ is constrained only by the requirement to give valid probabilities.

Choose $\eta$ such that $\sum_{\lambda:J_\lambda\leq J} p(\lambda)=(1-\omega)+\eta$, then $\eta\geq 0$. Set $\mathbf{E}^J:=\sum_{\lambda:J_\lambda\leq J}\tilde p(\lambda) \mathbf{E}_\lambda$, where $\tilde p(\lambda):=p(\lambda)/(1-\omega+\eta)$. Due to convexity of the set $\mathcal{Q}_{J,\alpha}$, we have $\mathbf{E}^J\in\mathcal{Q}_{J,\alpha}$. If $\omega-\eta=0$, set $\mathbf{E}'=(1,1)$, and otherwise, set $\mathbf{E}':=\sum_{\lambda:J_\lambda>J} p'_\lambda \mathbf{E}_\lambda$, where $p'_\lambda:=p(\lambda)/(\omega-\eta)$. Again, due to convexity, $\mathbf{E}'$ is a valid correlation. So we have
\begin{align}
\bm{E}^\omega&=(1-\omega)\mathbf{E}^J+\eta\mathbf{E}^J+(\omega-\eta)\mathbf{E}'\nonumber\\
&=(1-\omega)\mathbf{E}^J+\omega\mathbf{E}'',
\end{align}
where we go to the second line by defining $\bm{E}'':=\frac{\eta}{\omega}\bm{E}^J+(1-\frac{\eta}{\omega})\bm{E}'$, which is also a valid correlation, as it is the convex combination of two valid correlations.

In order to plot the set $\mathcal{Q}^\omega_{J,\alpha}$, we use the convention of \cite{VanHimbeeck2017}:
\begin{eqnarray}
g(E_1,E_2)&\coloneqq  &\frac{1}{2}\left(\sqrt{1\!+\!E_1}\sqrt{1\!+\!E_2}+\sqrt{1\!-\!E_1}\sqrt{1\!-\!E_2}\right)\nonumber\\
    &\geq &\gamma.\label{g(E1,E2)}
\end{eqnarray}
A correlation $\bm{E}^\omega$ belongs to $\mathcal{Q}^\omega_{J,\alpha}$ if it can be written in the form of (\ref{probabilistic J eqn}). Equivalently, for any correlation $\bm{E}^J$, we can get to a new point $\bm{E}^\omega$ as the result of mixing with an arbitrary $\bm{E}'$, for which we sometimes obtain a new correlation that is outside the original set, $\bm{E}^\omega\notin\mathcal{Q}_{J,\alpha}$. We can characteriz $\mathcal{Q}^\omega_{J,\alpha}$ by considering how $\bm{E}'$ can maximally violate $\mathcal{Q}_{J,\alpha}$. This is done by mixing all correlations where $E_1\leq E_2$ with $\bm{E}'=(-1,1)$, and mixing all correlations where $E_1\geq E_2$ with $\bm{E}'=(1,-1)$. In other words, we mix all correlations with the extremal corners $(\pm1,\mp1)$.
\begin{subequations}
\begin{align}
   E^\omega_1&=(1-\omega)E_1\mp\omega;  \\
   E^\omega_2&=(1-\omega)E_2\pm\omega.
\end{align}
\end{subequations}
Hence, in order to check whether some given correlation $\mathbf{E}^\omega$ is in $\mathcal{Q}_{J,\alpha}^\omega$ or not, one has to check whether
\begin{equation}
    \gamma\leq\begin{cases}
    g\Big(\frac{E^\omega_1+\omega}{1-\omega},\frac{E^\omega_2-\omega}{1-\omega}\Big), & \text{if}\:E_1^\omega\leq E_2^\omega; \\
    g\Big(\frac{E^\omega_1-\omega}{1-\omega},\frac{E^\omega_2+\omega}{1-\omega}\Big), & \text{if}\:E_1^\omega\geq E_2^\omega.
    \end{cases}
\end{equation}

\textit{Example: Coherent states.}
The relaxed quantum set bears relevance for instances in which the physical systems sent from preparation to measurement device satisfy our spin bound only approximately. For example, coherent states
\begin{equation}
    \ket{\beta}=e^{-\frac{|\beta|^2}{2}}\sum_{n=0}^\infty \frac{\beta^n}{\sqrt{n!}}\ket{n}\qquad (\beta\in\mathbb{C})
\end{equation}
contain superpositions of arbitrary photon numbers $n$. While we cannot exactly impose a constraint on the maximum spin (i.e.\ helicity or photon number), the coherent state can be approximated by the state
\begin{equation}
    \ket{\beta'}:=\frac{\Pi_N\ket{\beta}}{\sqrt{\bra{\beta}\Pi_N\ket{\beta}}}=\frac{\Pi_N\ket{\beta}}{\sqrt{1-\eta}}
\end{equation}
for which the constraint on the maximum spin $J=N$ holds exactly. Here, $\Pi_N=\sum_{n=0}^N |n\rangle\langle n|$ is the projector onto the subspace of photon numbers less than or equal to $N$, and we have defined 
\begin{align}
\bra{\beta}\Pi_N\ket{\beta}&=:1-\eta,
\end{align}
where $\eta$ is given by
\begin{equation}
    \eta=1-e^{-|\beta|^2}\sum^N_{n=0}\frac{|\beta|^{2n}}{n!}.
\end{equation}
Thus we can find the overlap as
\begin{equation}
    \braket{\beta}{\beta'}=\sqrt{1-\eta}.
\end{equation}
We can interpret $\eta$ as the probability that our constraint on $J=N$ does not hold, even though we do not assume that the measurement is actually performed in our setup. We would like to show that the coherent state $|\beta\rangle$ produces correlations that are close to the correlation set $\mathcal{Q}_{J,\alpha}$, given that it has large overlap with $|\beta'\rangle$. To this end, we will use the trace distance~\cite{nielsen2010} for quantum states
\begin{equation}
    D(\rho,\sigma)=\frac{1}{2}\mbox{tr}(|\rho-\sigma|)=\max_{0\leq M \leq \mathbf{1}}|{\rm tr}(M(\rho-\sigma))|.
\end{equation}
For pure states, the trace distance is given by $D(\ketbra{\phi}{\phi},\ketbra{\psi}{\psi})=\sqrt{1- |\braket{\phi}{\psi}|^2}$ and hence
\begin{equation}
    D(\ketbra{\beta}{\beta},\ketbra{\beta'}{\beta'})=\sqrt{\eta}.
\end{equation}

Therefore
\begin{eqnarray}
   |\bra{\beta}M_b\ket{\beta}-\bra{\beta'}M_b\ket{\beta'}|
   &\leq&\max_{0\leq M\leq\mathbf{1}} |{\rm tr}(M(|\beta\rangle\langle\beta|-|\beta'\rangle\langle \beta'|))|\nonumber\\
   &=&D(\ketbra{\beta}{\beta},\ketbra{\beta'}{\beta'})\nonumber\\
   &=&\sqrt{\eta}.
\end{eqnarray}
The same inequality will hold for measurements on the states $U_\alpha|\beta\rangle$ versus $U_\alpha |\beta'\rangle$. Thus, denoting the resulting correlations when $|\beta\rangle$ resp.\ $|\beta'\rangle$ is sent by $\mathbf{E}$ resp.\ $\mathbf{E}'$, we have
\begin{subequations}
    \begin{align}
        E_1=\langle\beta|(M_1-M_{-1})|\beta\rangle,\\ E'_1=\langle\beta'|(M_1-M_{-1})|\beta'\rangle,
    \end{align}
\end{subequations}
and thus $|E_1-E_1'|\leq 2 \sqrt{\eta}$. Similar argumentation shows that $|E_2-E'_2|\leq 2\sqrt{\eta}$. Note that this result can equivalently be posed in terms of probabilities as $|P^+_i-P_i^{+'}|\leq\sqrt{\eta}$ with $i\in\{0,\alpha\}$.
This allows us to determine a relaxed quantum set $\mathcal{Q}_{J,\alpha}^\delta$ that contains all the correlations generated by the coherent state $|\beta\rangle$.

To do so, we consider a general error $0\leq\kappa<1$ (above, $\kappa=\sqrt{\eta}$) and ask for the smallest $\delta\geq 0$ such that
\begin{equation}
   |E_i-E_i'|\leq 2\kappa, \mathbf{E}\in\mathcal{Q}_{J,\alpha}\Rightarrow \mathbf{E}'\in\mathcal{Q}_{J,\alpha}^\delta.
\end{equation}

In Fig.~\ref{errorPlot} we illustrate the set generated by this error, i.e.\ by the left-hand side of this implication. The plot is given in terms of probabilities $\mathbf{P}^+=(P^+_0,P^+_\alpha)=(P(+|0),P(+|\alpha))$, which are connected to the correlations $\mathbf{E}$ by the map
\begin{equation}
    \mathbf{P}^+=\frac{1}{2}\mathbf{E}+\left(\frac{1}{2},\frac{1}{2}\right),
    \label{EmapstoP}
\end{equation}
which is a bijective affine transformation, with an affine inverse
\begin{equation}
    \mathbf{E}=2 \mathbf{P}^+-(1,1).
\end{equation}  
The orange curves corresponding to the boundary of the error set given by $\kappa$ has been obtained by adding an error box around the curves $c_1,c_2$, which parametrize the extreme points of the quantum set $\cq_{J,\alpha}$ (besides the extremal points belonging to $\mathbf{P}^+=(0,0)$ and $\mathbf{P}^+=(1,1)$), given by \begin{eqnarray}
   c_1(\tau)&=& (\cos^2(J\tau),\cos^2(J(\tau+\alpha)))\quad(\tau\in I_1),\label{c1}\\
   c_2(\tau)&=&  (\cos^2(J\tau),\cos^2(J(\tau-\alpha)))\quad (\tau\in I_2),
   \label{c2}
\end{eqnarray}
where $I_1= \left[ 0,\frac \pi {2J}-\alpha\right]$ and $I_2=\left[\alpha,\frac\pi {2J}\right]$.

We start by looking at the point in the orange line such that $P^{+}_{\alpha}=0$. By geometrical arguments, this point can be found to be
\begin{equation}
\left(\cos^2(\arccos{\sqrt{\kappa}}-J\alpha)+\kappa,0\right).
\end{equation}
Then, from this point it follows that the set $\mathcal{Q}_{J,\alpha}^{\delta}$ has
\begin{equation}
\delta=2\left(\kappa+\sqrt{\kappa(1-\kappa)}\tan(J\alpha)\right)=2\tan(J\alpha)\sqrt{\kappa}+\mathcal{O}(\kappa).
\end{equation}
In Fig.~\ref{errorPlot} we show the boundary (green curve) of the relaxed quantum set $\mathcal{Q}_{J,\alpha}^{\delta}$. Most importantly, one observes that $\mathcal{Q}_{J,\alpha}^{\delta}$ includes the orange curve. While a proof remains to be found, we have strong numerical evidence to believe that this inclusion holds for any values of $J,\alpha,\kappa$.

This shows that our definition of the relaxed quantum set $\mathcal{Q}_{J,\alpha}^\delta$ also describes physically well-motivated examples like photon number unncertainty in single-mode coherent states, where $\delta=\mathcal{O}(\eta^{1/4})$, for $\eta$ the probability of observing more than $J=N$ photons.

\begin{figure}
\centering 
\includegraphics[width=0.5\columnwidth]{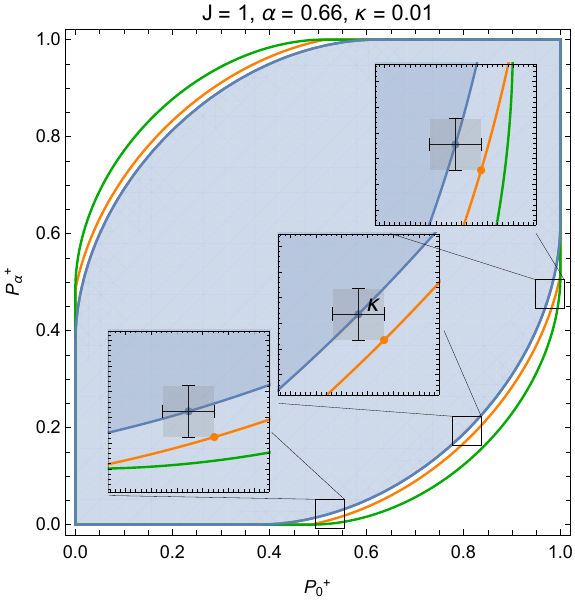}
\caption{The quantum set $\mathcal{Q}_{J,a}$ (blue), the boundary of the set with error $\kappa$ (orange) and the boundary of the smallest relaxed quantum set $\mathcal{Q}_{J,a}^{\delta}$ (green) such that it includes the previous error set given by $\kappa$. The plot illustrates that $\mathcal{Q}_{J,a}^{\delta}$ always contains the set given by the error $\kappa$.}
\label{errorPlot}
\end{figure}

\textit{Relaxed classical set.} Suppose now that the experimental error $\delta$ is of such a type that it could in principle be anticipated by an eavesdropper, such that there are more classical strategies available to her. We denote such \textit{epistemic} error by $\varepsilon$, and characteriz the set of classical correlations $\mathcal{C}^{\varepsilon}_{J,\alpha}$ accounting for this uncertainty. As in \cite{VanHimbeeck2017}, we define the relaxed classical set in terms of its quantum equivalent, with the additional assumption of deterministic outcomes:
\begin{align}
   \mathcal{C}_{J,\alpha}^\varepsilon=&\left\{\mathbf{E}=\sum_\lambda p(\lambda) \mathbf{E}_\lambda\,\,\left|\,\, \mathbf{E}_\lambda\in\mathcal{Q}_{J_\lambda,\alpha},
\sum_{\lambda:J_\lambda\leq J} p(\lambda)\geq 1-\varepsilon, \mathbf{E}_\lambda\in\{\pm 1\}\times \{\pm 1\}\right.\right\}.
\end{align}
Suppose that $J$ and $\alpha$ are such that $\mathcal{C}_{J,\alpha}$ is not the full square, $(\pm 1,\mp 1)\not\in\mathcal{C}_{J,\alpha}$. Then for all $\lambda$ with $J_\lambda\leq J$, we have $|E_{1|\lambda}-E_{2|\lambda}|=0$, and so all $\mathbf{E}^\varepsilon\in\mathcal{C}_{J,\alpha}^\varepsilon$ satisfy
\begin{equation}
|E^\varepsilon_1-E^\varepsilon_2|\leq\sum_\lambda p(\lambda) |E_{1|\lambda}-E_{2|\lambda}|
\leq 2\varepsilon. \label{quantum classical set}
\end{equation}
Conversely, suppose that $\mathbf{E}^\varepsilon$ satisfies Eq.~(\ref{quantum classical set}). Consider the case $E_1^\varepsilon\leq E_2^\varepsilon$ (the other case can be treated analogously). In this case, it is geometrically clear that the line starting at the corner $(-1,1)$ which crosses $\mathbf{E}^\varepsilon$ hits the diagonal at some point $(E,E)$ with $-1\leq E \leq 1$. In other words, $\mathbf{E}^\varepsilon=\kappa(-1,1)+(1-\kappa)(E,E)$ for some $\kappa\in[0,1]$. From this, we obtain $\kappa=\frac 1 2 (E_2^\varepsilon-E_1^\varepsilon)\leq\varepsilon$, and so $E^\varepsilon\in\mathcal{C}_{J,\alpha}^\varepsilon$, since $(E,E)$ is a convex combination of the deterministic correlations $(\pm 1,\pm 1)$.

In summary, inequality~(\ref{quantum classical set}) characterizes $\mathcal{C}_{J,\alpha}^\varepsilon$ exactly. \\

\section{Calculating $H^\star_{0,0,\alpha}$}
\label{pironiolink}
Here, we point out an equivalence between formulating the optimisation problem in Eq.~(\ref{optimisation})and the one presented in~\cite{van2019correlations}, in order to use the algorithm presented in~\cite[Sec.\ 3.3]{van2019correlations} to find a numerical solution to our optimisation problem.
In particular, the equivalence holds for the no-errors case (setting $\varepsilon=\omega=0$ in Eq.~(\ref{optimisation}), or (in the terminology of~\cite{van2019correlations}) no max-average assumption. Under these cases, the only difference between formulating both optimisation problems is on how the set of allowed quantum correlations is defined. In our case, we have
\begin{align}
\cq_{J,\alpha}:=\{ \textbf{E}|E_1=\Tr (\rho_1 M), \, E_2=\Tr (U_{\alpha} \rho_1 U_{\alpha}^{\dagger} M) \}.
\end{align} 
Meanwhile, in their case, the set of quantum correlations is defined as
\begin{align}
\mathcal{Q}_{w_{\text{pk}}}:=\{(\textbf{E},\textbf{w})|\textbf{w}\leq \textbf{w}_{\text{pk}}, \, E_x=\Tr (\rho_x M), \, \Tr (\rho_x \mathcal{O})\leq w_x\},
\end{align} 
where $x=\{1,2\}$, $\mathcal{O}$ is some energy operator, and $w_1, w_2$ gives an upper bound on the energy peak (i.e.\ a ``max-peak'' assumption). In both cases, the set of correlations is characterizd by the overlap between quantum states. In particular, we have the inequality
\begin{align}
\frac{1}{2}\left(\sqrt{1+E_1}\sqrt{1+E_2}+\sqrt{1-E_1}\sqrt{1-E_2}\right)\geq \left\{\begin{matrix}
\cos(J\alpha)&\mbox{if}& |J\alpha|< \frac{\pi}{2}\\
0 &\mbox{if}&|J\alpha|\geq\frac{\pi}{2}
\end{matrix}\right. ,
\end{align}
whilst, in~\cite{van2019correlations}, they have
\begin{align}
\frac{1}{2}\left(\sqrt{1+E_1}\sqrt{1+E_2}+\sqrt{1-E_1}\sqrt{1-E_2}\right)\geq \sqrt{1-w_1}\sqrt{1-w_2}-\sqrt{w_1}\sqrt{w_2}  .
\end{align}
In our case, in order to calculate $H^\star_{0,0,\alpha}$, one can perform the algorithm in~\cite[Sec.\ 3.3]{van2019correlations} by setting some $w_1=w_2=w_{\text{pk}}$ such that $\cos(J \alpha)=1-2w_{\text{pk}}$. A numerical study on the amount of randomness is beyond the scope of this manuscript. Nonetheless, the curious reader may be referred to Fig.~$5$(b) of \cite{van2019correlations} for an instance of $H^\star_{0,0,\alpha}$, for some given $J$ and $\alpha$ fulfilling the constraints.

\section{Bounding $H^\star_{\varepsilon,\omega,\alpha}$}
\label{boundinghstar}
Note that we have $\mathcal{Q}_{J,\alpha}=\mathcal{Q}_{J,-\alpha}$, and thus $\mathcal{Q}_{J,\alpha}^\omega=\mathcal{Q}_{J,-\alpha}^\omega$ and $H^\star_{\varepsilon,\omega,\alpha}=H^\star_{\varepsilon,\omega,-\alpha}$. Thus, it is sufficient to consider the case $\alpha>0$ in this section.

The relaxed quantum and classical sets play a role for characterising imperfect protocols, in which we may want to consider types of error introduced into our setup, whereby our assumption on $J$ may fail. First, we have $\varepsilon$-type error, which can be thought of approximately as \textit{epistemic} uncertainty of the experimenters; there may be additional variables $\lambda$ to which they may not have access (but to which Eve may), such that $\lambda$ describes exactly when our assumption fails. Second, we have $\omega$-type error, which can be thought of as genuine, \textit{ontic} randomness, such as that introduced by sending coherent states, for which the photon number (and the spin) can only be described probabilistically. This type of randomness cannot be described by any classical side-information that Eve may hold, and even she could never predict when the assumption fails. The amount of randomness $H^\star=H^\star_{\varepsilon,\omega,\alpha}$ that can be certified using our setup is thus affected by the room that we wish to grant for these types of experimental errors. 

To bound the certified randomness given such errors, let $\{p(\lambda),\mathbf{E}^\lambda\}$ be a minimizing ensemble for $H_{\varepsilon,\omega,\alpha}^\star$, i.e.
\begin{eqnarray}
H_{\varepsilon,\omega,\alpha}^\star &=& \min_{\{p(\lambda),\mathbf{E}^\lambda\}} \sum_\lambda p(\lambda) H(\mathbf{E}^\lambda),\label{eqConstraint1}\\
&&\mbox{subject to }\sum_{\lambda:\mathbf{E}^\lambda \in \mathcal{Q}_{J,\alpha}^\omega} p(\lambda)\geq 1-\varepsilon,\label{eqConstraint2}\\
&&\mbox{and }\sum_\lambda p(\lambda)\mathbf{E}^\lambda = \mathbf{E}.\label{eqConstraint3}
\end{eqnarray}
Then, for all $\varepsilon'\geq\varepsilon$, we trivially have
\begin{equation}
   \sum_{\lambda:\mathbf{E}^\lambda \in\mathcal{Q}_{J,\alpha}^\omega}p(\lambda)\geq 1-\varepsilon\geq 1-\varepsilon',
\end{equation}
and so the ensemble $\{p(\lambda),\mathbf{E}^\lambda\}$ satisfies the conditions that define the optimisation problem for $H^\star_{\varepsilon',\omega,\alpha}$. Consequently,
\begin{equation}
   H^\star_{\varepsilon',\omega,\alpha}\leq H^\star_{\varepsilon,\omega,\alpha}.
   \label{EqStar1}
\end{equation}
Similarly, $\omega'\geq\omega$ implies $\mathcal{Q}_{J,\alpha}^{\omega'}\supseteq \mathcal{Q}_{J,\alpha}^\omega$, and so
\begin{equation}
   \sum_{\lambda:\mathbf{E}^\lambda \in\mathcal{Q}_{J,\alpha}^{\omega'}}p(\lambda)\geq \sum_{\lambda:\mathbf{E}^\lambda \in\mathcal{Q}_{J,\alpha}^{\omega}}p(\lambda)\geq 1-\varepsilon,
\end{equation}
hence $\{p(\lambda),\mathbf{E}^\lambda\}$ is also a candidate ensemble for $H^\star_{\varepsilon,\omega',\alpha}$. It follows that
\begin{equation}
H^\star_{\varepsilon,\omega',\alpha}\leq H^*_{\varepsilon,\omega,\alpha}.
\label{EqStar2}
\end{equation}
In particular, Eqs.~(\ref{EqStar1}) and~(\ref{EqStar2}) imply the following result:
\begin{lemma}
Both types of error $\varepsilon,\omega>0$ decrease the number of certified random bits:
\begin{equation}
H^\star_{\varepsilon,\omega,\alpha}\leq H^\star_{0,0,\alpha}.
\end{equation}
\end{lemma}
To obtain inequalities in the converse direction, the following intermediate result will be useful. It is motivated by the Taylor expansion
\begin{equation}
   \arccos\left(\strut(1-\omega)^2 \cos x\right)=x+2 \omega\cot x+\mathcal{O}(\omega^2),
\end{equation}
which also shows that the constant factor appearing in front of $\omega$ in the following lemma cannot be improved.
\begin{lemma}
\label{LemArccos}
Suppose that $0\leq\omega\leq 1$ and $0\leq x \leq\frac\pi 2$. Then
\begin{equation}
   (1-\omega)^2\cos x \geq \cos_*(x+2\omega\cot x),
\end{equation}
where
\begin{equation}
   \cos_*(t):=\left\{
      \begin{array}{cl}
      	  \cos t &\mbox{if } 0\leq t \leq \frac \pi 2\\
      	  0 & \mbox{otherwise}.
      \end{array}
   \right.
\end{equation}
\end{lemma}
\begin{proof}
If $x+2\omega\cot x \geq \frac\pi 2$, then the inequality is trivially true, since the left-hand side is non-negative. Thus, it is sufficient to consider the case $x+2\omega\cot x < \frac\pi 2$. For every fixed $x$, consider the function $g:[0,1]\to\mathbb{R}$, given by
\begin{equation}
   g(\omega):=\sin(x+2\omega\cot x)-(1-\omega)\sin x.
\end{equation}
We have $g(0)=0$ and
\begin{equation}
   g'(\omega)=2\cos(x+2\omega\cot x)\cot x + \sin x \geq 0.
\end{equation}
Therefore, $g(\omega)\geq 0$ for all $\omega\in[0,1]$. Now consider
\begin{equation}
   f(\omega):=(1-\omega)^2 \cos x -\cos(x+2\omega\cot x).
\end{equation}
We have $f(0)=0$ and
\begin{equation}
   f'(\omega)=2(\cot x) g(\omega)\geq 0.
\end{equation}
Thus $f(\omega)\geq 0$ for all $\omega\in [0,1]$.
\end{proof}
\begin{lemma}
\label{LemInclusionRelaxed}
For $0\leq\omega\leq 1$, we have the set inclusion
\begin{equation}
   \mathcal{Q}_{J,\alpha}^\omega \subseteq \mathcal{Q}_{J,\alpha+2\omega\cot(J\alpha)/J}.
\end{equation}
\end{lemma}
\begin{proof}
If $J\alpha+2\omega\cot(J\alpha)>\frac\pi 2$, then $\mathcal{Q}_{J,\alpha+2\omega\cot(J\alpha)/J}=[-1,1]^2$, and the claim is trivially true. Thus, we may assume that $J\alpha+2\omega\cot(J\alpha)\leq\frac\pi 2$. Suppose that $\mathbf{E}\in\mathcal{Q}_{J,\alpha}^\omega$, then there exist $\mathbf{E}^J\in\mathcal{Q}_{J,\alpha}$ and $\mathbf{E}'\in[-1,1]^2$ such that
\begin{equation}
   \mathbf{E}=(1-\omega)\mathbf{E}^J+\omega \mathbf{E}'.
\end{equation}
Now we use the fact the the functions $x\mapsto \sqrt{1-x}$ and $x\mapsto \sqrt{1+x}$ are concave and non-negative on $[-1,1]$:
\begin{eqnarray}
&&\frac 1 2  \left(\sqrt{1+E_1}\sqrt{1+E_2}+\sqrt{1-E_1}\sqrt{1-E_2}\right)\nonumber\\
&=&\frac 1 2 \left(\sqrt{1+(1-\omega)E_1^J+\omega E'_1}\sqrt{1+(1-\omega)E_2^J+\omega E'_2}\right.\nonumber\\
&&\left. \quad + \sqrt{1-\left(\strut(1-\omega)E_1^J+\omega E'_1\right)}\sqrt{1-\left(\strut(1-\omega)E_2^J+\omega E'_2\right)}\right)\nonumber\\
&\geq& \frac 1 2 \left[\left((1-\omega)\sqrt{1+E_1^J}+\omega\sqrt{1+E'_1}\right)\left((1-\omega)\sqrt{1+E_2^J}+\omega\sqrt{1+E'_2}\right)\right.\nonumber\\
&& \left. \quad+ \left((1-\omega)\sqrt{1-E_1^J}+\omega\sqrt{1-E'_1}\right)\left((1-\omega)\sqrt{1-E_2^J}+\omega\sqrt{1-E'_2}\right)\right]\nonumber\\
&\geq& \frac 1 2 \left[ \sqrt{1+E_1^J}\sqrt{1+E_2^J}+\sqrt{1-E_1^J}\sqrt{1-E_2^J}\right](1-\omega)^2\nonumber\\
&\geq& (1-\omega)^2 \cos(J\alpha)\nonumber\\
&\geq& \cos_*(J\alpha +2\omega \cot(J\alpha)),
\end{eqnarray}
where we have used Lemma~\ref{LemArccos} in the final step. Therefore, $\mathbf{E}\in\mathcal{Q}_{J,\alpha+2\omega\cot(J\alpha)/J}$, and the claim follows.
\end{proof}
Now recall Eqs.~(\ref{eqConstraint1}), (\ref{eqConstraint2}) and~(\ref{eqConstraint3}) which are satisfied by the minimizing ensemble $\{p(\lambda),\mathbf{E}^\lambda\}$ for $H^\star_{\varepsilon,\omega,\alpha}$. In particular, Eq.~(\ref{eqConstraint2}) and Lemma~\ref{LemInclusionRelaxed} imply that
\begin{equation}
   \sum_{\lambda:\mathbf{E}^\lambda\in\mathcal{Q}_{J,\alpha+2\omega\cot(J\alpha)/J}}p(\lambda)\geq    \sum_{\lambda:\mathbf{E}^\lambda\in\mathcal{Q}_{J,\alpha}^\omega}p(\lambda)\geq 1-\varepsilon,
   \label{eqConstraint4}
\end{equation}
and so Eqs.~(\ref{eqConstraint3}) and~(\ref{eqConstraint4}) show that $\{p(\lambda),\mathbf{E}^\lambda\}$ is a candidate ensemble for $H^\star_{J,\alpha+2\omega\cot(J\alpha)/J}$. This proves the following corollary:
\begin{corollary}
\label{CorEpsZero}
We have
\begin{equation}
   H^\star_{\varepsilon,\omega,\alpha}\geq H^\star_{\varepsilon,0,\alpha+2\omega \cot(J\alpha)/J}.
\end{equation}
\end{corollary}
Finally, we would like to obtain an inequality that tells us what happens if replace a finite value of $\varepsilon$ by zero. To this end, we will use an inverse concavity property (see e.g.~\cite{nielsen2010}) of Shannon entropy
\begin{equation}
   H(q)=H(q_1,\ldots,q_n):=-\sum_{i=1}^n q_i \log q_i.
\end{equation}
If we have $n$ probability distributions $p^{(i)}=(p^{(i)}_1,\ldots,p^{(i)}_m)$ (vectors with non-negative entries summing to unity), and another probability distribution $q=(q_1,\ldots,q_n)$, then
\begin{equation}
H\left(\sum_{i=1}^ n q_i p^{(i)}\right) \leq \sum_{i=1}^n q_i H(p^{(i)})+H(q).
\end{equation}
In the notation of the main text, the corresponding entropy of a correlation $\mathbf{E}$ therefore {satisfies}
\begin{eqnarray}
   H(t\mathbf{E}+(1-t)\mathbf{E}')\leq t H(\mathbf{E})+(1-t)H(\mathbf{E}')-t\log t -(1-t)\log(1-t)
   \label{eqConcavity}
\end{eqnarray}
for all $0\leq t\leq 1$ (with the convention $0\log 0:=0$). To see this, use for example that $H(\mathbf{E})=\frac 1 2 \sum_x H(p^{(x)}_{\mathbf{E}})$, where $p^{(x)}_{\mathbf{E}}$ is the binary distribution with probabilities $\frac{1\pm E_x}2$. We will use this to prove the following result:
\begin{lemma}
\label{LemInverseConcavity}
For every $0\leq\varepsilon,\omega< 1$, we have
\begin{equation}
   H^\star_{\varepsilon,\omega,\alpha}\geq H^\star_{0,\varepsilon+\omega,\alpha}+\log(1-\varepsilon)-\frac{\varepsilon\log(2/\varepsilon)}{1-\varepsilon},
\end{equation}
with the convention that $0\log(2/0):=0$.	
\end{lemma}
\begin{proof}
Consider again a minimizing ensemble for $H^\star_{\varepsilon,\omega,\alpha}$, satisfying Eqs.~(\ref{eqConstraint1})--(\ref{eqConstraint3}). Such choice of ensemble also entails a definition of a set of ``hidden variables'' $\Lambda\ni\lambda$. Define the subset $\Lambda_\omega:=\{\lambda\in\Lambda\,\,|\,\,\mathbf{E}^\lambda \in\mathcal{Q}_{J,\alpha}^\omega\}$, then
\begin{equation}
   t:=1-\sum_{\lambda\in\Lambda_\omega}p(\lambda)\leq\varepsilon.
\end{equation}
Therefore, $\mathbf{E}^J:=\sum_{\lambda\in\Lambda_\omega}\frac{p(\lambda)}{1-t}\mathbf{E}^\lambda\in\mathcal{Q}_{J,\alpha}^\omega$ because $\mathcal{Q}_{J,\alpha}^\omega$ is convex. Now, if $t=0$, let $\mathbf{E}'$ be an arbitrary correlation, while for $t>0$, set $\mathbf{E}':=\sum_{\lambda\not\in\Lambda_\omega}\frac{p(\lambda)} t \mathbf{E}^\lambda$, then $\mathbf{E}=(1-t)\mathbf{E}^J+t\mathbf{E}'$. For all $\lambda\in\Lambda_\omega$, set $q(\lambda):=p(\lambda)/(1-t)$ and $\mathbf{F}^\lambda :=(1-t)\mathbf{E}^\lambda+t\mathbf{E}'$. Then direct calculation shows that
\begin{equation}
   \sum_{\lambda\in\Lambda_\omega} q(\lambda)\mathbf{F}^\lambda = (1-t)\mathbf{E}^J+t\mathbf{E}'=\mathbf{E}.
\end{equation}
Since $\mathbf{E}^\lambda\in\mathcal{Q}_{J,\alpha}^\omega$, there exist correlations $\mathbf{E}^{\lambda,J}\in \mathcal{Q}_{J,\alpha}$ and $\mathbf{E}^{\lambda,\bullet}\in [-1,1]^2$ such that $\mathbf{E}^\lambda=(1-\omega)\mathbf{E}^{\lambda,J}+\omega \mathbf{E}^{\lambda,\bullet}$. Thus
\begin{eqnarray}
\mathbf{F}^\lambda&=& (1-t)\mathbf{E}^\lambda + t \mathbf{E}'\nonumber\\
&=&(1-t)(1-\omega)\mathbf{E}^{\lambda,J}+(1-t)\omega \mathbf{E}^{\lambda,\bullet}+t\mathbf{E}'\nonumber\\
&\in& \mathcal{Q}_{J,\alpha}^{t+\omega-t\omega}\subseteq \mathcal{Q}_{J,\alpha}^{t+\omega}\subseteq \mathcal{Q}_{J,\alpha}^{\varepsilon+\omega}.
\end{eqnarray}
It follows that $\{q(\lambda),\mathbf{F}^\lambda\}_{\lambda\in\Lambda_\omega}$ is a candidate ensemble for $H^\star_{0,\varepsilon+\omega,\alpha}$. Using furthermore that $H(\mathbf{E}')\leq \log 2$, we obtain
\begin{eqnarray}
H^\star_{0,\varepsilon+\omega,\alpha}&\leq& \sum_{\lambda\in\Lambda_\omega} q(\lambda) H(\mathbf{F}^\lambda)\nonumber\\
&\leq& \sum_{\lambda\in\Lambda} q(\lambda)H\left((1-t)\mathbf{E}^\lambda + t \mathbf{E}'\right)\nonumber\\
&\leq& \sum_{\lambda\in\Lambda} \frac{p(\lambda)}{1-t}\left((1-t)H(\mathbf{E}^\lambda)+t H(\mathbf{E}')-t\log t -(1-t)\log(1-t)\right)\nonumber\\
&=&\sum_{\lambda\in\Lambda}p(\lambda) H(\mathbf{E}^\lambda)+\frac {tH(\mathbf{E}')} {1-t} -\frac{t\log t}{1-t}-\log(1-t)\nonumber\\
&\leq& H^\star_{\varepsilon,\omega,\alpha}+\frac{t\log(2/t)}{1-t}-\log(1-t).
\end{eqnarray}
Using finally that $t\leq\varepsilon$ completes the proof.
\end{proof}
Applying Lemma~\ref{LemInverseConcavity} and Corollary~\ref{CorEpsZero} in succession shows the following:
\begin{corollary}
If $\varepsilon,\omega\in [0,1)$, then
\begin{equation}
   H^\star_{\varepsilon,\omega,\alpha}\geq H^\star_{0,0,\alpha+2(\varepsilon+\omega)\cot(J\alpha)/J}+\log(1-\varepsilon)-\frac{\varepsilon\log(2/\varepsilon)}{1-\varepsilon}.
\end{equation}
\end{corollary}

\section{Proof that $\cq_{J}\subseteq \crr_{J}$}
\label{proofquantumsubseteqrotationset}

Since $\mathcal{R}_J$ is a convex set, it is sufficient to show that extremal correlations of $\mathcal{Q}_J$ are contained in it. That is, we can disregard shared randomness and consider a fixed POVM $\{M_b\}_b$ and an arbitrary, normalized pure state
\begin{equation}
\ket{\phi}=\sum_{j=-J}^J \phi_j \ket{\psi_j},
\end{equation}
where $\ket{\psi_j}\in\ch_j$ is normalizd and $\phi_j=\braket{\psi_j}{\phi}$. We choose an orthonormal basis for $\mathcal{H}$, such that every $\ket{\psi_k}$ is an element of this basis, and we calculate the probabilities in this basis:
\begin{eqnarray}
P(b|\alpha)&=&\tr[U_\alpha\ketbra{\phi}{\phi}U_\alpha^\dagger M_b]\nonumber\\
&=&\sum_{j,j'=-J}^{J} \phi_j \phi_{j'}^* e^{ij\alpha} M^b_{jj'} e^{-ij'\alpha}\nonumber\\
 &=&\sum_{l=-2J}^{2J}\sum_{\begin{matrix}
-J\leq j,j'\leq J\!:\\
j\!-\!j'=l
\end{matrix}} \phi_j\phi_{j'}^*M^b_{jj'}e^{i(j-j')\alpha}\nonumber\\
 &=&\sum_{l=-2J}^{2J} a_l e^{il\alpha},
\end{eqnarray}
where we have defined the coefficients
 \begin{equation}
a_l=\sum_{\begin{matrix}
 -J\leq j,j'\leq J\!:\\
j\!-\!j'=l
 \end{matrix}} \phi_j\phi_{j'}^*M^b_{jj'},
\end{equation}
for $-J\leq l\leq J$, and $M^b_{jj'}:=\langle j'|M_b|j\rangle$. The coefficients have the property
\begin{eqnarray}
a_l&=&\sum_{\begin{matrix}
 -J\leq j,j'\leq J\\
 j\!-\!j'=l
 \end{matrix}} \phi_j \phi_{j'}^*M^b_{jj'}\nonumber\\
 &=&\sum_{\begin{matrix}
 -J\leq j,j'\leq J\\
j\!-\!j'=l
 \end{matrix}} \left(\phi_j^*\phi_{j'}\right)^*\left(M^b_{j'j}\right)^*=a^*_{-l},
\end{eqnarray}
which we can use to write 
\begin{eqnarray}
P(b|\alpha)&=&\sum_{l=-2J}^{2J} a_l (\cos(l\alpha)+i\sin(l\alpha))\nonumber\\
&=&a_0+\sum_{l=1}^{2J}\left( 2\Re(a_l)\cos(l\alpha)-2\Im(a_l)\sin(l\alpha)\right),\label{QuantumSBs}
\end{eqnarray}
 By observing that this is exactly of the form~(\ref{SBs}), we conclude that $\mathcal{Q}_J\subseteq\mathcal{R}_J$, and thus also $\mathcal{Q}_{J,\alpha}\subseteq \mathcal{R}_{J,\alpha}$.

\section{Proof that $\cq_{1/2}=\crr_{1/2}$}
\label{proofq1/2=r1/2}

In this section we will show that for a $J=1/2$ system, not only do the sets $\cq_{1/2,\alpha}$ and $\crr_{1/2,\alpha}$ coincide, but in fact every rotation box can be simulated by a quantum model, i.e. $\cq_{1/2}=\crr_{1/2}$.

\textit{Rotation boxes.} We will start our discussion by giving a characterisation of the convex set of rotation box correlations $\crr_{1/2}$.
 Every $J=1/2$ rotation box is described by probability distributions of the form
\begin{equation}
P(+1|\alpha)=c_0+c_1 \cos{\alpha}+ s_1 \sin\alpha,\label{probability2levelspacetimebox}
\end{equation} 
where $c_0,c_1,s_1\in\mathbb{R}$, s.t. $0\leq P(+1|\alpha)\leq 1\,\forall \alpha$. Conditions for $c_0,c_1$ and $s_1$ can be found by the fact that $P(+1|\alpha)$ has to produce valid probabilities for all $\alpha$ and if $P(+1|\alpha)=1$ or $P(+1|\alpha)=0$ it has to be a maximum or minimum respectively. This allows us to write the set $\crr_{1/2}$ in the form:
\begin{eqnarray}
\crr_{1/2}=\bigg\{P(+1|\alpha)\,\,|\,\,0\leq c_0 \leq 1;\,\, &\mbox{if}&\,\, \frac{1}{2}< c_0, \,\mbox{then}\, (c_0)^2\geq (c_1)^2+(s_1)^2;\nonumber\\ \,\, &\mbox{if}& \,\,  c_0\leq\frac{1}{2} \,\mbox{then}\, (1-c_0)^2\geq (c_1)^2+(s_1)^2    
\bigg\}.
\end{eqnarray}
The convexity of $\crr_{1/2}$ follows immediately from the fact that the set of probabilities is convex and that a convex combination of functions of the form $c_0+c_1\cos\alpha+s_1\sin\alpha$ gives again a function of the same form. 
Furthermore, the set $\crr_{1/2}$ is isomorphic to the subset $\mathcal{A}\subset \mathbb{R}^3$, given by
\begin{eqnarray}
\crr_{1/2}&\cong&\left\{\begin{pmatrix}
(x & y & z)
\end{pmatrix}^T| \,
0\leq x \leq 1;\,\mbox{if} \, x\leq\frac{1}{2}, \mbox{then}\, x^2\geq y^2+z^2;\,\mbox{if}\, \frac{1}{2}< x, \mbox{then}\, (1-x)^2\geq y^2+z^2    
\right\}\nonumber\\
&=& \mathcal{A}\subset \mathbb{R}^3.
\end{eqnarray}
\begin{figure}
\centering 
\includegraphics[width=.3\columnwidth]{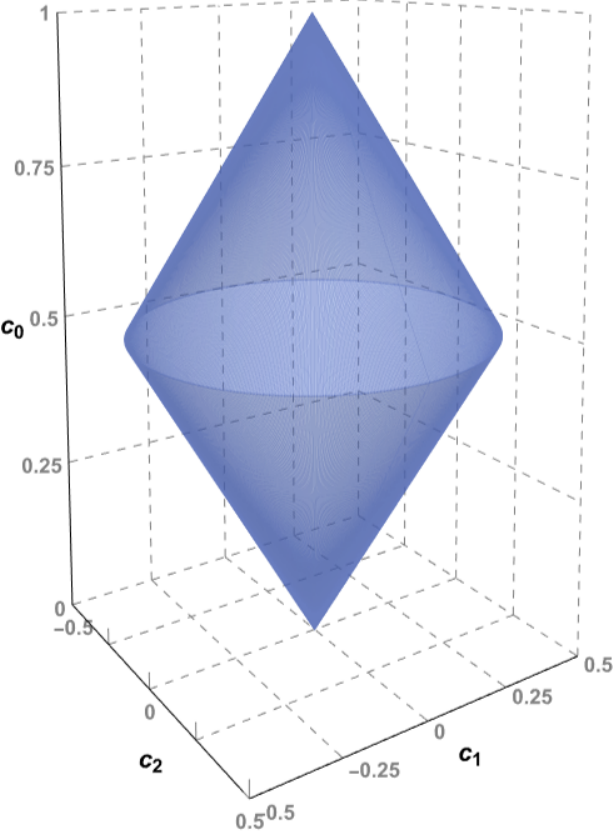}
\caption{$\crr_{1/2}=\cq_{1/2}$ as a convex set.}
\label{3dplot}
\end{figure}

 The isomorphism  $\varphi:S\rightarrow\mathcal{A}$ is given by
 \begin{equation}
 \varphi(P(+1|\alpha))=\begin{pmatrix}
 c_0\\ c_1 \\ s_1
 \end{pmatrix},
 \end{equation}
 which is clearly linear and therefore affine -- that is, $\mathcal{A}$ inherits its convex properties form $\crr_{1/2}$ and vice versa.

It can be  shown that the extreme points $\partial_{\rm ext}\mathcal{A}$ of the set $\mathcal{A}$, and therefore the extreme points $\partial_{\rm ext}\crr_{1/2}$ of $\crr_{1/2}$ are given by:
\begin{eqnarray}
\partial_{\rm ext}\mathcal{A}&=&\left\{\begin{pmatrix}
0\\
0\\
0
\end{pmatrix},\begin{pmatrix}
1\\
0\\
0
\end{pmatrix},\left\{\left.\begin{pmatrix}
\frac{1}{2}\\
y\\
z
\end{pmatrix}\right| y^2+z^2=\frac{1}{4}\right\}\right\}\nonumber\\
&\cong& \partial_{\rm ext}\crr_{1/2}.
\end{eqnarray}
An illustration of $\crr_{1/2}$ is given in Fig.~\ref{3dplot}.

\textit{Rotation Boxes Simulated by Quantum Model.}
We now find quantum models for the extreme points of $\crr_{1/2}$.

The extreme points $P(+1|\alpha)=1$ and $P(+1|\alpha)=0$ are just constant probability distributions and therefore can obviously be simulated by quantum theory. For example, consider $\ket{\phi_{\frac{1}{2}}}\in \mathcal{H}_{+\frac{1}{2}}$ and the measurement  $\left\{P_+=\ketbra{\phi_{+\frac{1}{2}}}{\phi_{+\frac{1}{2}}},P_-=\mathbf{1}-P_+\right\}$. This yields the probability distribution\begin{eqnarray}
P(+1|\alpha)&=&\bra{\phi_{+\frac{1}{2}}}U^\dagger_{\alpha}P_+U_{\alpha}\ket{\phi_{+\frac{1}{2}}}\nonumber\\
&=&\left|e^{i\frac{\alpha}{2}}\braket{\phi_{+\frac{1}{2}}}{\phi_{+\frac{1}{2}}}\right|^2=1.
\end{eqnarray}
In similar fashion one can construct a quantum model for $P(+1|\alpha)=0$.

All other extreme points are of the form 
\begin{eqnarray}
P_{\tau}(+1|\alpha)&=&\frac{1}{2}+c_1^\tau \cos\alpha+s_1^\tau\sin\alpha\nonumber\\
&=&\frac{1}{2}\left(1+\cos\alpha\cos\tau+\sin\alpha\sin\tau\right)\nonumber\\
&=&\frac{1}{2}\left(1+\cos(\alpha-\tau)\right),\label{distextrempoints}
\end{eqnarray}
where we have used that we can write $c^\tau_1=(1/2)\cos\tau$ and $s_1^\tau=(1/2)\sin\tau$ such that $\left(c^\tau_1\right)^2+\left(s_1^\tau\right)^2=1/4$ is satisfied.

To find a quantum model for these probability distributions we define the state \begin{equation}
\ket{\phi}=\frac{1}{\sqrt{2}}\left(\ket{\phi_{+\frac{1}{2}}}+\ket{\phi_{-\frac{1}{2}}}\right),
\end{equation}
where $\ket{\phi_{\pm\frac{1}{2}}}\in\mathcal{H}_{\pm\frac{1}{2}}$.

Furthermore, we will consider the family of states\begin{equation}
\ket{\phi_\tau}=U_{-\tau}\ket{\phi}
\end{equation}
and the projective measurement $\{P_+=\ketbra{\phi}{\phi},P_-=\mathbf{1}-P_+\}$.
Hence, we find the probability distributions
\begin{eqnarray}
P_{\tau}(+1|\alpha)&=&\bra{\phi_{\tau}}U^\dagger_\alpha P_+ U_{\alpha}\ket{\phi_\tau}=\bra{\phi}U^\dagger_{\alpha-\tau}P_+ U_{\alpha-\tau}\ket{\phi}\nonumber\\
&=&\left|\bra{\phi}U_{\alpha-\tau}\ket{\phi}\right|^2=\left|\frac{1}{2}\left(e^{i\frac{\alpha-\tau}{2}}+e^{-i\frac{\alpha-\tau}{2}}\right)\right|^2\nonumber\\
&=& \cos^2\left(\frac{\alpha-\tau}{2}\right)=\frac{1}{2}\left(1+\cos(\alpha-\tau)\right),
\end{eqnarray}
which are precisely the same probability distributions as in (\ref{distextrempoints}).

We have seen that all extreme points of $\crr_{1/2}$ can be associated with probability distributions that are compatible with $\mathcal{Q}_{1/2}$. To also associate the non-extreme points of $\crr_{1/2}$ with probabilities compatible with $\mathcal{Q}_{1/2}$, we use that all (non-extreme) points $P$ can be written as a convex combination of of at most four extreme points due to Carath\'eodory's theorem~\cite{Webster}. That is,
\begin{equation}
P(+|\alpha)=\sum_{i=1}^4\lambda^i P^i_{ex}(+|\alpha),
\end{equation}
with $0\leq\lambda^i\leq 1$ and $\sum_{i=1}^4 \lambda^i=1$. As just discussed, we can write down  a quantum model for all extreme points
\begin{equation}
P^i_{ex}(+1|\alpha)=\bra{\phi^i}U_{\alpha}^\dagger P^i_+ U_{\alpha}\ket{\phi^i}.
\end{equation}   
Now, we introduce a four-dimensional auxiliary system  and we define the state
\begin{equation}
\ket{\psi}=\sum_{i=1}^4\sqrt{\lambda^i} \ket{\phi^i}\otimes\ket{i},
\end{equation}
where $\ket{\phi^i}$ are the states from the quantum model that generated the extreme points.
Furthermore, we adapt the representation of ${\rm SO}(2)$ to $U_\alpha\otimes\mathbf{1}$, which does not affect the assumption on the maximal spin $J$ (see also Appendix~\ref{purestateeqmixedstateset}).
 We define the following measurement operator
 \begin{equation}
 P_+=\sum_{i=1}^4 P^i_+\otimes \ketbra{i}{i}.
\end{equation}  
Then a simple calculation gives
\begin{eqnarray}
\bra{\psi}\big(U_\alpha^\dagger\otimes\mathbf{1}\big)P_+\big(U_\alpha\otimes\mathbf{1}\big)\ket{\psi}=P(+|\alpha)
\end{eqnarray}
for all $\alpha$, which shows that we can also always find quantum models that are compatible with $\mathcal{Q}_{1/2}$ for all non-extreme points of $\crr_{1/2}$. Hence, $\crr_{1/2}=\cq_{1/2}$. In~\cite{Aloy}, an alternative proof of  $\crr_{1/2}=\cq_{1/2}$ is provided by introducing a GPT system capable of reproducing all correlations of $\crr_{1/2}$ and then showing the equivalence of this GPT system with the rebit.

\section{Proof that $\cq_{J,\alpha}=\crr_{J,\alpha}$}
\label{proofquantumeqrotationset}

We want to show that the rotation box correlations given by Eq.~(\ref{rotation_box_set}) describes the same set as that in Eq.~(\ref{quantum_set}), for the quantum box. We can show this by rearranging Eq.~(\ref{rotation_box_set}):
\begin{align}
    \cos{J\alpha} &\leq \cos\Big(\frac{1}{2}\big|\arcsin E_2-\arcsin E_1\big|\Big) \nonumber\\
    &= \cos\Big(\frac{1}{2}\arcsin E_2-\frac{1}{2}\arcsin E_1\Big) \nonumber\\
    &= \cos\big(\frac{1}{2}\arcsin E_2\big)\cos\big(\frac{1}{2}\arcsin E_1\big) + \sin\big(\frac{1}{2}\arcsin E_2\big)\sin\big(\frac{1}{2}\arcsin E_1\big).
\end{align}
We use the identities $\cos x=\sqrt{\frac{1}{2}+\frac{1}{2}\cos(2x)}$ and $\sin x=\sqrt{\frac{1}{2}-\frac{1}{2}\cos(2x)}$, and then for the following line $\cos(\arcsin{x})=\sqrt{1-x^2}$. Therefore the above can be further rewritten as follows:
\begin{align}
    =& \sqrt{\frac{1}{2}+\frac{1}{2}\cos\big(\arcsin E_2\big)}\sqrt{\frac{1}{2}+\frac{1}{2}\cos\big(\arcsin E_1\big)} +\sqrt{\frac{1}{2}-\frac{1}{2}\cos\big(\arcsin E_2\big)}\sqrt{\frac{1}{2}-\frac{1}{2}\cos\big(\arcsin E_1\big)}\nonumber\\
    =& \sqrt{\frac{1}{2}+\frac{\sqrt{1-E^2_2}}{2}}\sqrt{\frac{1}{2}+\frac{\sqrt{1-E^2_1}}{2}} +\sqrt{\frac{1}{2}-\frac{\sqrt{1-E^2_2}}{2}}\sqrt{\frac{1}{2}-\frac{\sqrt{1-E^2_1}}{2}}.
\end{align}

We can derive the following useful identity:
\begin{align}
    \sqrt{\frac{1}{2}\pm\frac{\sqrt{1\!-\!x^2}}{2}}
=\frac{1}{2}\big(\sqrt{1+x}\pm\sqrt{1-x}\big),
\end{align}
which can be used as a substitute so that we arrive at:
\begin{align}
    \cos{J\alpha} &\leq \frac{1}{4}\bigg(\! \big(\sqrt{1\!+\!E_2}+\sqrt{1\!-\!E_1}\big) \big(\sqrt{1\!+\!E_2}+\sqrt{1\!-\!E_1}\big) + \big(\sqrt{1\!+\!E_2}-\sqrt{1\!-\!E_1}\big) \big(\sqrt{1\!+\!E_2}-\sqrt{1\!-\!E_1}\big)\! \bigg) \nonumber\\
    &= \frac{1}{2}\Big(\sqrt{1+E_2}\sqrt{1+E_1} + \sqrt{1-E_2}\sqrt{1-E_1}\Big).
\end{align}
The final line is the same as that given in Eq.~(\ref{quantum_set}); i.e.\ the rotation box condition described by Eq.~(\ref{rotation_box_set}) is identical to the quantum condition described by Eq.~(\ref{quantum_set}). Therefore, combining with Appendix~\ref{proofquantumsubseteqrotationset}, we have shown that $\crr_{J,\alpha}=\mathcal{Q}_{J,\alpha}$.


\begin{thebibliography}{99}



\bibitem{Hardy1}
 
L.\ Hardy, \textit{Quantum Theory From Five Reasonable Axioms}, arXiv:quant-ph/0101012 (2001),
\href{https://doi.org/10.48550/arXiv.quant-ph/0101012}{DOI:10.48550/arXiv.quant-ph/0101012}.

 
\bibitem{vanDam}
W.\ van Dam, \textit{Implausible consequences of superstrong nonlocality}, Nat.\ Comput.\ \textbf{12}, 9--12 (2013),
\href{https://doi.org/10.1007/s11047-012-9353-6}{DOI:10.1007/s11047-012-9353-6}.
 
\bibitem{Navascues}
M.\ Navascu\'es and H.\ Wunderlich, \textit{A glance beyond the quantum model}, Proc.\ R.\ Soc.\ Lond.\ A \textbf{466}, 881--890 (2009),
\href{https://doi.org/10.1098/rspa.2009.0453}{DOI:10.1098/rspa.2009.0453}.
 
\bibitem{DakicBrukner}
B.\ Daki\'c and \v{C}.\ Brukner, \textit{Quantum theory and beyond: is entanglement special?}, in Deep Beauty: Understanding the Quantum World through Mathematical Innovation (ed.\ H.\ Halvorson), Cambridge University Press, Cambridge (2011),
\href{https://doi.org/10.1017/CBO9780511976971.011}{DOI:10.1017/CBO9780511976971.011}.
 
\bibitem{Chiribella}
G.\ Chiribella, G.\ M.\ d'Ariano, and P.\ Perinotti, \textit{Informational derivation of quantum theory}, Phys.\ Rev.\ A \textbf{84}, 012311 (2011),
\href{https://doi.org/10.1103/PhysRevA.84.012311}{DOI:10.1103/PhysRevA.84.012311}.
 
\bibitem{Masanes}
L.\ Masanes and M.\ P.\ M\"uller, \textit{A derivation of quantum theory from physical requirements}, New J.\ Phys.\ \textbf{13}(6), 063001 (2011),
\href{https://doi.org/10.1088/1367-2630/13/6/063001}{DOI:10.1088/1367-2630/13/6/063001}.
 
\bibitem{Popescu}
S.\ Popescu, \textit{Nonlocality beyond quantum mechanics}, Nat.\ Phys.\ \textbf{10}, 264--270 (2014), \href{https://doi.org/10.1038/nphys2916}{DOI:10.1038/nphys2916}.
 
\bibitem{Mueller}
M.\ P.\ M\"uller, \textit{Probabilistic Theories and Reconstructions of Quantum Theory}, SciPost Phys.\ Lect.\ Notes \textbf{28} (2021),
\href{https://doi.org/10.21468/SciPostPhysLectNotes.28}{DOI:10.21468/SciPostPhysLectNotes.28}.
 
\bibitem{Plavala}
M.\ Pl\'avala, \textit{General probabilistic theories: An introduction}, Phys.\ Rep.\ \textbf{1033}, 1--64 (2023), \href{https://doi.org/10.1016/j.physrep.2023.09.001}{DOI:10.1016/j.physrep.2023.09.001}.
 
\bibitem{mayers1998quantum}
D.\ Mayers and A.\ Yao, \textit{Quantum cryptography with imperfect apparatus}, Proceedings 39th Annual Symposium on Foundations of Computer Science (IEEE, 1998), 503--509, \href{https://doi.org/10.1109/SFCS.1998.743501}{DOI:10.1109/SFCS.1998.743501}.
 
\bibitem{barrett2005no}
J.\ Barrett, L.\ Hardy, A.\ Kent, \textit{No signaling and quantum key distribution}, Phys.\ Rev.\ Lett.\ \textbf{95}, 010503 (2005),
\href{https://doi.org/10.1103/PhysRevLett.95.010503}{DOI:10.1103/PhysRevLett.95.010503}.
 
\bibitem{colbeck}
R.\ Colbeck, \textit{Quantum And Relativistic Protocols For Secure Multi-Party Computation}, PhD thesis, University of Cambridge, 2006, \href{https://doi.org/10.48550/arXiv.0911.3814}{DOI:10.48550/arXiv.0911.3814}.
 
\bibitem{acin2007device}
A.\ Ac{\'\i}n, N.\ Brunner, N.\ Gisin, S.\ Massar, S.\ Pironio, and V.\ Scarani, \textit{Device-independent security of quantum cryptography against collective attacks}, Phys.\ Rev.\ Lett.\ \textbf{98}, 230501 (2007), \href{https://doi.org/10.1103/PhysRevLett.98.230501}{DOI:10.1103/PhysRevLett.98.230501}.
 
\bibitem{gallego2010device}
R.\ Gallego, N.\ Brunner, C.\ Hadley, and A.\ Ac{\'\i}n, \textit{Device-independent tests of classical and quantum dimensions}, Phys.\ Rev.\ Lett.\ \textbf{105}, 230501 (2010), \href{https://doi.org/10.1103/PhysRevLett.105.230501}{DOI:10.1103/PhysRevLett.105.230501}.
 
\bibitem{pawlowski2011semi}
M.\ Paw{\l}owski, and N.\ Brunner, \textit{Semi-device-independent security of one-way quantum key distribution}, Phys.\ Rev.\ A \textbf{84}, 010302 (2011), \href{https://doi.org/10.1103/PhysRevA.84.010302}{DOI:10.1103/PhysRevA.84.010302}.
 
\bibitem{liang2011semi}
Y.\-C.\ Liang, T.\ V{\'e}rtesi, and N.\ Brunner, \textit{Semi-device-independent bounds on entanglement}, Phys.\ Rev.\ A \textbf{83}, 022108 (2011), \href{https://doi.org/10.1103/PhysRevA.83.022108}{DOI:10.1103/PhysRevA.83.022108}.
 
\bibitem{branciard2012one}
C.\ Branciard, E.\ Cavalcanti, S.\ Walborn, V.\ Scarani, and H.\ M.\ Wiseman, \textit{One-sided device-independent quantum key distribution: Security, feasibility, and the connection with steering}, Phys.\ Rev.\ A \textbf{85}, 010301 (2012), \href{https://doi.org/10.1103/PhysRevA.85.010301}{DOI:10.1103/PhysRevA.85.010301}.
 
\bibitem{VanHimbeeck2017}
T.\ Van Himbeeck, E.\ Woodhead, N.\ J.\ Cerf, R.\ Garc\'ia-Patr\'on, and S.\ Pironio, \textit{Semi-device-independent framework based on natural physical assumptions}, Quantum \textbf{1}, 33 (2017), \href{https://doi.org/10.22331/q-2017-11-18-33}{DOI:10.22331/q-2017-11-18-33}.
 
\bibitem{brunner2014bell}
N.\ Brunner, D.\ Cavalcanti, S.\ Pironio, V.\ Scarani, and S.\ Wehner, \textit{Bell nonlocality}, Rev.\ Mod.\ Phys.\ \textbf{86}, 419 (2014), \href{https://doi.org/10.1103/RevModPhys.86.419}{DOI:10.1103/RevModPhys.86.419}.
 
\bibitem{scarani2019bell}
V.\ Scarani, \textit{Bell nonlocality}, Oxford University Press, Oxford, 2019, \href{https://doi.org/10.1093/oso/9780198788416.001.0001}{DOI:10.1093/oso/9780198788416.001.0001}.
 
\bibitem{li2011semi}
H.-W.\ Li, Z.-Q.\ Yin, Y.-C.\ Wu, X.-B.\ Zou, S.\ Wang, W.\ Chen, G.-C.\ Guo, and Z.-F.\ Han, \textit{Semi-device-independent random number expansion without entanglement}, Phys.\ Rev.\ A \textbf{84}, 034301 (2011), \href{https://doi.org/10.1103/PhysRevA.84.034301}{DOI:10.1103/PhysRevA.84.034301}.
 
\bibitem{acin2016certified}
A.\ Ac{\'\i}n and L.\ Masanes, \textit{Certified randomness in quantum physics}, Nature \textbf{540}, 213 (2016), \href{https://doi.org/10.1038/nature20119}{DOI:10.1038/nature20119}.
 
\bibitem{ma2016quantum}
X.\ Ma, X.\ Yuan, Z.\ Cao, B.\ Qi, and Z.\ Zhang \textit{Quantum random number generation}, npj Quantum Inf.\ \textbf{2}, 1 (2016), \href{https://doi.org/10.1038/npjqi.2016.21}{DOI:10.1038/npjqi.2016.21}.
 
\bibitem{rusca2019self}
D.\ Rusca, T.\ van Himbeeck, A.\ Martin, J.\ B.\ Brask, W.\ Shi, S.\ Pironio, N.\ Brunner, and H.\ Zbinden, \textit{Self-testing quantum random number generator based on an energy bound}, Phys.\ Rev.\ A \textbf{100}, 062338 (2019), \href{https://doi.org/10.1103/PhysRevA.100.062338}{DOI:10.1103/PhysRevA.100.062338}.
 
\bibitem{tebyanian2021semi}
H.\ Tebyanian, M.\ Zahidy, M.\ Avesani, A.\ Stanco, P.\ Villoresi, and G.\ Vallone, \textit{Semi-device independent randomness generation based on quantum state's indistinguishability}, Quantum Sci.\ Technol.\ \textbf{6}, 045026 (2021), \href{https://doi.org/10.1088/2058-9565/ac2047}{DOI:10.1088/2058-9565/ac2047}.
 
\bibitem{van2019correlations}
T.\ van Himbeeck and S.\ Pironio, \textit{Correlations and randomness generation based on energy constraints}, arXiv:1905.09117, \href{https://doi.org/10.48550/arXiv.1905.09117}{DOI:10.48550/arXiv.1905.09117}.

 
\bibitem{bowles2014certifying}
J.\ Bowles, M.\ T.\ Quintino, and N.\ Brunner, \textit{Certifying the dimension of classical and quantum systems in a prepare-and-measure scenario with independent devices}, Phys.\ Rev.\ Lett.\ \textbf{112}, 140407 (2014), \href{https://doi.org/10.1103/PhysRevLett.112.140407}{DOI:10.1103/PhysRevLett.112.140407}.
 
\bibitem{brunner2008testing}
N.\ Brunner, S.\ Pironio, A.\ Ac{\'\i}n, N.\ Gisin, A.\ A.\ M{\'e}thot, and V.\ Scarani, \textit{Testing the Dimension of Hilbert Spaces}, Phys.\ Rev.\ Lett.\ \textbf{100}, 210503 (2008), \href{https://doi.org/10.1103/PhysRevLett.100.210503}{DOI:10.1103/PhysRevLett.100.210503}.
 
\bibitem{Aloy}
A.\ Aloy, T.\ D.\ Galley, C.\ L.\ Jones, S.\ L.\ Ludescher, and M.\ P.\ M\"uller, \textit{Spin-Bounded Correlations: Rotation Boxes Within and Beyond Quantum Theory}, Commun.\ Math.\ Phys.\ \textbf{405}, 292 (2024), \href{https://doi.org/10.1007/s00220-024-05123-2}{DOI:10.1007/s00220-024-05123-2}.
 
\bibitem{MuellerGarner}
M.\ P.\ M\"uller and A.\ J.\ P.\ Garner, \textit{Testing Quantum Theory by Generalizing Noncontextuality}, Phys.\ Rev.\ X \textbf{13}, 041001 (2023), \href{https://doi.org/10.1103/PhysRevX.13.041001}{DOI:10.1103/PhysRevX.13.041001}.
 
\bibitem{pironio2016focus}
S.\ Pironio, V.\ Scarani, and T.\ Vidick, \textit{Focus on device independent quantum information}, New J.\ Phys.\ \textbf{18}, 100202 (2016), \href{https://doi.org/10.1088/1367-2630/18/10/100202}{DOI:10.1088/1367-2630/18/10/100202}.
 
\bibitem{Pauwels}
J.\ Pauwels, A.\ Tavakoli, E,\ Woodhead, and S.\ Pironio, \textit{Entanglement in prepare-and-measure scenarios: many questions, a few answers}, New J.\ Phys.\ \textbf{24}, 063015 (2022), \href{https://doi.org/10.1088/1367-2630/ac724a}{DOI:10.1088/1367-2630/ac724a}.
 
\bibitem{Wald}
R.\ Wald, \textit{General Relativity}, Chicago University Press, Chicago, 1984, \href{https://doi.org/10.7208/chicago/9780226870373.001.0001}{DOI:10.7208/chicago/9780226870373.001.0001}.
 
\bibitem{Garner2020}
A.\ J.\ P.\ Garner, M.\ Krumm, and M.\ P.\ M\"uller, \textit{Semi-device-independent information processing with spatiotemporal degrees of freedom}, Phys.\ Rev.\ Research \textbf{2}, 013112 (2020), \href{https://doi.org/10.1103/PhysRevResearch.2.013112}{DOI:10.1103/PhysRevResearch.2.013112}.
 
\bibitem{caban2003photon}
P.\ Caban, and J.\ Rembieli{\'n}ski, \textit{Photon polarization and Wigner's little group}, Phys.\ Rev.\ A \textbf{68}, 042107 (2003), \href{https://doi.org/10.1103/PhysRevA.68.042107}{DOI:10.1103/PhysRevA.68.042107}.
 
\bibitem{Peres}
A.\ Peres, \textit{Quantum Theory: Concepts and Methods}, Kluwer Acad.\ Publ., Dordrecht, 2010, \href{https://doi.org/10.1007/0-306-47120-5}{DOI:10.1007/0-306-47120-5}.
 
\bibitem{DeVore}
R.\ A. DeVore and G.\ G.\ Lorentz, \textit{Constructive Approximation}, Springer Verlag, Berlin, Heidelberg, 1993.
 
\bibitem{tavakoli2022informationally}
A.\ Tavakoli, E.\ Z.\ Cruzeiro, E.\ Woodhead, and S.\ Pironio, \textit{Informationally restricted correlations: a general framework for classical and quantum systems}, Quantum \textbf{6}, 620 (2022), \href{https://doi.org/10.22331/q-2022-01-05-620}{DOI:10.22331/q-2022-01-05-620}.
 
\bibitem{tavakoli2020informationally}
A.\ Tavakoli, E.\ Z.\ Cruzeiro, J.\ B.\ Brask, N.\ Gisin, and N.\ Brunner, \textit{Informationally restricted correlations}, Quantum \textbf{4}, 332 (2020), \href{https://doi.org/10.22331/q-2020-09-24-332}{DOI:10.22331/q-2020-09-24-332}.
 
\bibitem{Barrett}
J.\ Barrett, \textit{Information processing in generalized probabilistic theories}, Phys.\ Rev.\ A \textbf{75}, 032304 (2007), \href{https://doi.org/10.1103/PhysRevA.75.032304}{DOI:10.1103/PhysRevA.75.032304}.
 
\bibitem{Webster}
R.\ Webster, \textit{Convexity}, Oxford University Press, Oxford, 1994, \href{https://doi.org/10.1093/oso/9780198531470.001.0001}{DOI:10.1093/oso/9780198531470.001.0001}.
 
\bibitem{hall2013}
B.\ C.\ Hall, \textit{Quantum Theory for Mathematicians}, Graduate Texts in Mathematics, Springer International Publishing (2013), \href{https://doi.org/10.1007/978-1-4614-7116-5}{DOI:10.1007/978-1-4614-7116-5}.
 
\bibitem{hatcher2002}
A.\ Hatcher \textit{Algebraic Topology}, Cambridge University Press (2002), \href{https://doi.org/10.1017/S0013091503214620}{DOI:10.1017/S0013091503214620}.
 
\bibitem{hall2015}
B.\ C.\ Hall, \textit{Lie Groups, Lie Algebras, and Representations: An Elementary Introduction}, Graduate Texts in Mathematics Springer International Publishing (2015), \href{https://doi.org/10.1007/978-3-319-13467-3}{DOI:10.1007/978-3-319-13467-3}.
 
\bibitem{stone1932}
M.\ H.\ Stone, \textit{On one-parameter unitary groups in Hilbert space}, Ann.\ Math.\ \textbf{33}(3), 643-648 (1932), \href{https://doi.org/10.2307/1968538}{DOI:10.2307/1968538}.
 
\bibitem{Wightman}
A.\ S.\ Wightman, \textit{Superselection Rules: Old and New}, Nuovo Cimento B \textbf{110}, 751 (1995), \href{https://doi.org/10.1007/BF02741478}{DOI:10.1007/BF02741478}.
 
\bibitem{nielsen2010}
M.\ Nielsen, I.\ Chuang, \textit{Quantum Computation and Quantum Information: 10th Anniversary Edition}, Cambridge University Press (2010), \href{https://doi.org/10.1017/CBO9780511976667}{DOI:10.1017/CBO9780511976667}.
 
\bibitem{grimmett2020probability}
G.\ Grimmett, D.\ Stirzaker, \textit{Probability and random processes}, Oxford University Press, Oxford (2020).
%https://doi.org/10.1093/oso/9780198572237.001.0001
 
\bibitem{pironio2013security}
S.\ Pironio, S.\ Massar, \textit{Security of practical private randomness generation}, Phys.\ Rev.\ A \textbf{87}(1), 012336 (2013), \href{https://doi.org/10.1103/PhysRevA.87.012336}{DOI:10.1103/PhysRevA.87.012336}.
 
\bibitem{zhang2018certifying}
Y.\ Zhang, E.\ Knill, and P.\ Bierhorst, \textit{Certifying quantum randomness by probability estimation}, Phys.\ Rev.\ A \textbf{98}(4), 040304 (2018), \href{https://doi.org/10.1103/PhysRevA.98.040304}{DOI:10.1103/PhysRevA.98.040304}.
 
\bibitem{knill2020generation}
E.\ Knill, Y.\ Zhang, P.\ Bierhorst, \textit{Generation of quantum randomness by probability estimation with classical side information}, Phys.\ Rev.\ Research \textbf{2}(3), 033465 (2020), \href{https://doi.org/10.1103/PhysRevResearch.2.033465}{DOI:10.1103/PhysRevResearch.2.033465}.

\end{thebibliography}
\end{document}